\documentclass[11pt]{article}
\usepackage{cite}
\usepackage{a4}
\usepackage{amsmath,amsfonts}
\usepackage{amssymb}
\usepackage{epic,gastex}
\usepackage{array}
\usepackage{amsthm}
\usepackage{hhline}
\usepackage{graphicx}
\usepackage{url}

\newcommand{\sa}{synchronizing automata}
\newcommand{\san}{synchronizing automaton}
\newcommand{\scn}{strongly connected}
\newcommand{\sw}{reset word}
\newcommand{\sws}{reset words}
\newcommand{\ssw}{reset word of minimum length}
\newcommand{\rl}{reset threshold}

\DeclareSymbolFont{rsfscript}{OMS}{rsfs}{m}{n} \DeclareSymbolFontAlphabet{\mathrsfs}{rsfscript}

\newcommand{\qqed}{\hfill$\Box$}

\newtheorem{theorem}{Theorem}

\newtheorem{prop}{Proposition}
\newtheorem{lemma}{Lemma}

\renewenvironment{proof}{\trivlist\item[\hskip\labelsep{\bf
Proof}.] }{\qqed\endtrivlist}

\newtheorem{conjecture}{Conjecture}

\begin{document}
\title{Primitive digraphs with large exponents\\ and slowly
synchronizing automata\thanks{A preliminary version of a part of the results of this paper was published in~\cite{AGV10}.}}

\author{D. S. Ananichev \and  V. V. Gusev \and M. V. Volkov}

\date{}

\maketitle

\begin{abstract}
We present several infinite series of \sa\ for which the minimum
length of reset words is close to the square of the number of
states. All these automata are tightly related to primitive
digraphs with large exponent.
\end{abstract}

\section{Background and the structure of the paper}

This paper has arisen from our attempts to find a theoretical explanation for the results of certain computational experiments in
synchronization of finite automata. Recall that a (\emph{complete deterministic}) \emph{finite automaton} (DFA) is a triple
$\mathrsfs{A}=\langle Q,\Sigma,\delta\rangle$, where $Q$ and $\Sigma$ are finite sets called the \emph{state set} and the \emph{input
alphabet} respectively, and $\delta:Q\times\Sigma\to Q$ is a totally defined function called the \emph{transition function}. As usual,
$\Sigma^*$ stands for the collection of all finite words over the alphabet $\Sigma$, including the empty word 1. The function $\delta$
extends to a function $Q\times\Sigma^*\to Q$ (still denoted by $\delta$) as follows: for every $q\in Q$ and $w\in\Sigma^*$, we set
$\delta(q,w)=q$ if $w=1$ and $\delta(q,w)=\delta(\delta(q,v),a)$ if $w=va$ for some $v\in\Sigma^*$ and $a\in\Sigma$. Thus, via $\delta$,
every word $w\in\Sigma^*$ acts on the set $Q$.

A DFA $\mathrsfs{A}=\langle Q,\Sigma,\delta\rangle$ is said to be
\emph{synchronizing} if some word $w\in\Sigma^*$ brings all states
to one particular state: $\delta(q,w)=\delta(q',w)$ for all
$q,q'\in Q$. Any such word $w$ is said to be a \emph{reset word}
for the DFA. The minimum length of reset words for $\mathrsfs{A}$
is called the \emph{\rl} of $\mathrsfs{A}$.

Synchronizing automata serve as transparent and natural models of error-resistant systems in many applied areas (system and protocol
testing, information coding, robotics). At the same time, \sa\ surprisingly arise in some parts of pure mathematics (symbolic dynamics,
theory of substitution systems and others). Basics of the theory of \sa\ as well as its diverse connections and applications are discussed,
for instance, in~\cite{Sa05,Vo08}. Here we focus on only one aspect of the theory, namely, on the question of how the \rl\ of a DFA depends
on the state number.

For brevity, a DFA with $n$ states will be referred to as an
$n$-\emph{automaton}. In~1964 \v{C}ern\'{y}~\cite{Ce64}
constructed a series of synchronizing $n$-automata with \rl\
$(n-1)^2$. Soon after that he conjectured that these automata
represent the worst possible case with respect to synchronization
speed, i.e.\ that every synchronizing $n$-automaton can be reset
by a word of length $(n-1)^2$. This hypothesis has become known as
the \emph{\v{C}ern\'{y} conjecture}. In spite of its simple
formulation and many researchers' efforts, the \v{C}ern\'{y}
conjecture remains unproved (and undisproved) for more than 45
years. Moreover, no upper bound of magnitude $O(n^2)$ for the \rl\
of a synchronizing $n$-automaton is known so far\footnote{Up
today, the best upper bound on the \rl\ of a synchronizing
$n$-automaton is the bound $\frac{n^3-n}6$ found by
Pin~\cite{Pi83} in 1983. A slightly better upper bound
$\frac{n(7n^2+6n-16)}{48}$ has been recently published
in\cite{Tr11} but the proof of this result contains an unclear
place.}.

Why is the \v{C}ern\'{y} conjecture so inaccessible? A detailed discussion of this important issue would go far beyond the scope of the
present paper but one of the difficulties encountered by the theory of \sa\ is worth registering here. We mean the shortage of examples of
\emph{extremal} automata, i.e.\ $n$-automata having \rl\ $(n-1)^2$. In fact, the series found in~\cite{Ce64} still remains the only known
infinite series of extremal automata. Besides that, we know only a few isolated examples of such automata, the largest (with respect to the
state number) being the 6-automaton discovered by Kari~\cite{Ka01} in 2001. (See~\cite{Vo08} for a complete list of known extremal
automata.) Moreover, even $n$-automata whose \rl\ is close to $(n-1)^2$ have been very rare in the literature so far---besides the
\v{C}ern\'{y} series one can only refer to the series from~\cite{AVZ}. With a very restricted number of examples, it was difficult to
verify various guesses and assumptions that arose when researchers were searching for approaches to the \v{C}ern\'{y} conjecture. That is
why the history of investigations in this area abounds in  ``false trails'', i.e.\ auxiliary hypotheses that looked promising at first but
were disproved after some time. (Cf.~\cite{Be11} for an analysis of a number of such ``false trails''.)

How can one find slowly \sa? Experiments (see, e.g., \cite{ST11}) demonstrate that with probability very close to~1, a random automaton is
reset by a word of length much less than the state number. Therefore it is impossible to encounter by chance an automaton whose \rl\ is
close to the square of state number, and one has to reveal such automata via exhaustive search. It was such an exhaustive search experiment
that served as a departure point of the present paper.

Our methodology and some results of the experiment are described in Section~2. We have noticed a similarity between the observable behavior
of the number of \sa\ with a fixed number of states as a function of their \rl\ and the well-studied behavior of the number of primitive
digraphs with a fixed number of vertices as a function of their exponent. We discuss this similarity in Section~3 after recalling the
necessary concepts and facts from the theory of primitive digraphs. The main results of the paper are collected in Section~4. We show that
slowly \sa\ revealed in our experiment represent some infinite series of such automata and that each of these series is tightly related to
a certain known series of primitive digraphs with large exponent. This connection between digraphs and automata allows us to provide
transparent and uniform proofs for all statements concerning the \rl\ of series of slowly \sa---both for the previously known series and
the series that have first appeared in the present paper. The proof technique is, to the best of our knowledge, new and appears to be of
independent interest. In Section~5 we discuss further prospects of the suggested  approach and formulate a few new conjectures.

\section{Methodology and some results of the experiment}

As mentioned in Section~1, finding automata whose \rl\ is close to
the square of state number requires an exhaustive search. Since
the quantity of $n$-automata drastically grows with $n$, such a
search should be designed in a reasonable way. For instance,
specifying a 9-automaton with two input letters is equivalent to
specifying a pair of function on a 9-element set. There are
$9^{18}\approx 1.50\times 10^{17}$ such pairs, and if one will
spend one nanosecond for calculating the \rl\ of each automaton
defined this way, the exhaustive search would take around five
years. Clearly, if an $n$-automaton with $k$ input letters is
specified by a $k$-tuple of function on an $n$-element set, each
particular automaton is generated $n!k!$ times. However it is not
possible to speed up the search by screening out isomorphic
automata---even for $n=9$ and $k=2$ neither time nor memory would
suffice to check whether the current automaton is isomorphic to
one of the previously generated automata.

In order to optimize search, we have employed a string representation of initially-connected automata suggested in~\cite{AMR}. Recall that
a DFA $\mathrsfs{A}=\langle Q,\Sigma,\delta\rangle$ is said to be \emph{initially-connected} (from a state $q_0$), if one can reach any
state from $q_0$ by applying a suitable word: for every $q\in Q$ there exists $w\in\Sigma^*$ such that $q=\delta(q_0,w)$. A DFA which is
initially-connected from each of its states is called \emph{\scn}. In general, a \san\ may fail to be \scn\ or initially-connected. However
it is well known that one can restrict to \scn\ automata when dealing with issues related to the \v{C}ern\'{y} conjecture. This is a
consequence of the following easy result.
\begin{prop}[\!{\mdseries\cite[Proposition~2.1]{Vo09}}]
\label{reduction} Let $\mathrsfs{A}=\langle
Q,\Sigma,\delta\rangle$ be a \san\ and let $S$ be the set of all
states to which $\mathrsfs{A}$ can be reset. Then
$\mathrsfs{S}=\langle S,\Sigma,\delta|_S\rangle$ is a \scn\
subautomaton in $\mathrsfs{A}$ and for every function
$f:\mathbb{Z}^+\to\mathbb{N}$ satisfying
$$f(n)\ge\frac{n(n-1)}2\ \text{ and }\ f(n)\ge f(n-m+1)+f(m) \text{ for } n\ge m\ge 1,$$
the fact that the \rl\ of the automaton $\mathrsfs{S}$ is bounded by $f(|S|)$ implies that the \rl\ for $\mathrsfs{A}$ does not exceed
$f(|Q|)$.
\end{prop}
In particular, taking $(n-1)^2$ as $f(n)$, we can conclude that if
$\mathrsfs{A}$ is a counterexample for the \v{C}ern\'{y}
conjecture, then so is the \scn\ subautomaton $\mathrsfs{S}$.
Similarly, if $\mathrsfs{A}$ has \rl\ close to the square of state
number, then so does $\mathrsfs{S}$. Thus, restricting our search
to initially-connected automata, we do not risk to overlook a
counterexample for the \v{C}ern\'{y} conjecture or any interesting
slowly \san.

Now we describe the string representation from~\cite{AMR}. Let a
DFA $\mathrsfs{A} = \langle Q, \Sigma, \delta\rangle$ be
initially-connected from a state $q_0$. We fix a linear ordering
of the input alphabet $\Sigma$ and traverse the DFA by
breadth-first search starting at $q_0$ and choosing the outgoing
transitions according to the ordering. Let the state $q_0$ have
number 0 and let all other states in $Q$ be numbered in the order
of their appearance in breadth-first search. For instance, the
states of the DFA in Fig.~\ref{fig:AMR} are numbered as follows:
\begin{center}
\begin{tabular}{c|c|c|c}
$A$ & $B$ & $C$ & $D$\\
\hline 0 & 2 & 1& 3\end{tabular}
\end{center}
provided that the input letters are ordered as $a<b<c$ and the state $A$ is chosen as $q_0$.
\begin{figure}[ht]
\begin{center}
\unitlength .8mm
\begin{picture}(84,45)(0,-60)
\gasset{Nw=11.0,Nh=11.0,Nmr=5.5}
\node[iangle=90.0,ilength=6.0,Nmarks=i](n1)(23.99,-24.0){$A$}
\node(n2)(60.49,-24.25){$C$} \node(n3)(59.9,-56.18){$D$}
\node(n6)(25.06,-56.08){$B$} \drawloop[loopangle=180.0](n1){$c$}
\drawloop[loopangle=0.0](n3){$b$} \drawloop[loopangle=180.0](n6){$c$}
\drawedge[curvedepth=3.6](n1,n2){$a$} \drawedge[curvedepth=3.6](n2,n1){$c$}
\drawedge[ELside=r,curvedepth=-3.6](n3,n2){$a$}
\drawedge[ELside=r,curvedepth=-3.6](n2,n3){$b$}
\drawedge[ELside=r,curvedepth=-3.6](n6,n1){$b$}
\drawedge[ELside=r,curvedepth=-3.6](n1,n6){$b$}
\drawedge[curvedepth=3.6](n6,n3){$a$} \drawedge[curvedepth=3.6](n3,n6){$c$}
\drawedge(n2,n6){$a$}
\end{picture}
\caption{The DFA with the canonical string
$[1,2,0,2,3,0,3,0,2,1,3,2]$}\label{fig:AMR}
\end{center}
\end{figure}
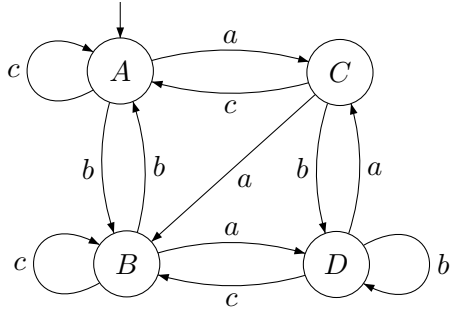

\noindent We assign a string of length  $|\Sigma|$ to each state
$q$ of $\mathrsfs{A}$; the $i$-th position of the string holds the
number of the state to which $q$ is sent under the action of the
$i$-th letter. If we concatenate all these strings in the
increasing order of the state numbers, we get a string of numbers
from the set $\{0,1,\dots,|Q|-1\}$ that has length $|Q||\Sigma|$
and uniquely determines $\mathrsfs{A}$. It is called the
\emph{canonical string} of the DFA $\mathrsfs{A}$. For instance,
the canonical string of the DFA in Fig.~\ref{fig:AMR} is
$[1,2,0,2,3,0,3,0,2,1,3,2]$.

It is not hard to see (cf.~\cite[Theorem~5]{AMR}) that a string
$[s_0,\dots,s_{nk-1}]$ of numbers from the set $\{0,1,\dots,n-1\}$
is a canonical string of an initially-connected $n$-automaton with
$k$ input letters if and only if the following two conditions are
satisfied:
\begin{enumerate}
\item[(R1)] for each $i$ such that $s_i>1$, there is $j<i$ with $s_j = s_i-1$,
\item[(R2)] for each $m$ such that $1\le m <n$, there is $j<mk$ with $s_j=m$.
\end{enumerate}

Using this observation, we calculated \rl{}s of
initially-connected $n$-automata with two input letters as
follows. We used a 128-core grid of AMD Opteron 2.6 GHz
processors. The grid belongs to the Institute of Mathematics and
Mechanics of the Ural Branch of the Russian Academy of Sciences;
it runs under Linux, has 256 Gb of memory and the peak performance
of 665.6 GFLOPS. One node of the gird generated relatively small
portions of strings satisfying (R1) and (R2) and sent them to
other nodes that worked on their portions of automata in parallel.
The management program was written in C with MPI. Standard
algorithms (cf.~\cite{Sa05,Vo08}) were implemented to test whether
the current automaton is synchronizing and to calculate its \rl.
Both implementations were written in C. We notice that the
synchronization test is very fast as it works on the digraph of
pairs of states while the calculation of \rl\ works on the digraph
of non-empty sets of states and in the worst case its running time
exponentially depends on the size of the automaton under
inspection\footnote{It is known \cite[Theorem~4]{OU10} that the
problem of computing the \rl\ of a given automaton is complete for
the functional analogue $\mathsf{FP}^\mathsf{NP[log]}$ of the
complexity class $\mathsf{P}^\mathsf{NP[log]}$ consisting of all
decision problems solvable by a deterministic polynomial-time
Turing machine that has an access to an oracle for an
\textsf{NP}-complete problem, with the number of queries being
logarithmic in the size of the input.}. However in practice the
\rl\ was calculated fairly fast since, as mentioned in Section~1,
it is small for an overwhelming majority of automata.

Presenting automata via their canonical strings drastically
reduces the exhaustive search. (It is easy to see, for instance,
that every $n$-automaton with 2 input letters is generated at most
$2n$ times if presented this way.) Nevertheless, the search still
remains quite large. For $n=9$, say, the number of automata to be
analyzed was 705\,068\,085\,303. However, thanks to
parallelization the computation for $n=9$ took less than 24 hours.

As the result of the computation, we built an array that, for each
possible value of \rl, contains the number of automata attaining
this value. A part of results that was of outmost importance for
us (namely, the part related to slowly \sa) had been
double-checked with the package TESTAS~\cite{Tr06}.

Table~\ref{automata} shows a part of the output array for the case
$n=9$. Here automata are counted up to isomorphism.
\begin{table}[ht]
\extrarowheight=1pt \tabcolsep=4.8pt \caption{Reset thresholds of
synchronizing 9-automata with two input letters} \label{automata}
{\small
\begin{tabular}{|p{2.9cm}||c|c|c|c|c|c|c|c|c|c|c|c|c|c|}
\hline
\centering{$N$} & 64 & 63 & 62 & 61 & 60 & 59 & 58 & 57 & 56 & 55 & 54 & 53 & 52 & 51 \\
\hline \raggedright{\# of automata with \rl\ $N$}
&\raisebox{-11pt}{1} &\raisebox{-11pt}{0} &\raisebox{-11pt}{0}
&\raisebox{-11pt}{0} &\raisebox{-11pt}{0} &\raisebox{-11pt}{0}
&\raisebox{-11pt}{1} &\raisebox{-11pt}{2} &\raisebox{-11pt}{3}
&\raisebox{-11pt}{0} &\raisebox{-11pt}{0} &\raisebox{-11pt}{0}
&\raisebox{-11pt}{4} &\raisebox{-11pt}{4} \\
\hline
\end{tabular}}
\end{table}

Clearly, the unique automaton in the column corresponding to $N=64$ is nothing but the 9-automaton from the \v{C}ern\'{y} series. Then one
observes a gap: no 9-automata with two input letters have \rl\ in the range from 59 to 63. This gap was mentioned by
Trahtman~\cite{Tr06,Tr06a} who reported that for $n=7,8,9,10$ no $n$-automata with two input letters and \rl\ in the range from $n^2-3n+5$
to $n^2-2n$ were registered in his experiments. The gap is followed by an ``island'' consisting of three values attained by 6 automata, and
the ``island'' is followed by yet another gap. To the best of our knowledge, the second gap has not been reported in the literature up to
now.

The behavior just described---a unique extremal value followed by a gap which in turn is followed by a small ``island'' and yet another gap
---persists also for automata with a larger number of states. The size of the ``island'' depends only on the parity of the state number
and the sizes of the gaps grow as linear functions of the state number. A similar behavior is known for another value investigated in
discrete mathematics, namely, for the number of primitive digraphs with a fixed vertex number and a given exponent. In the next section we
recall these notions and discuss the similarity in more detail.

\section*{{\centerline{\large\bf \S3. Primitive digraphs and their exponents}}}

A \emph{digraph} (directed graph) is a pair $D=\langle
V,E\rangle$, where $V$ is a finite set and $E\subseteq V\times V$.
The elements of $V$ and $E$ are called \emph{vertices} and
respectively \emph{edges}. We notice that this definition allows
loops but excludes multiple edges. If $v,v'\in V$ and $e=(v,v')\in
E$, then we say that the edge $e$ \emph{starts} at $v$. We assume
that the reader is acquainted with the basic notions of digraph
theory such as directed path, directed cycle, isomorphism etc.

If $D=\langle V,E\rangle$ is a digraph, then its \emph{incidence
matrix} (referred to simply as \emph{matrix} in the sequel) is a
$V\times V$-matrix whose entry in row $v$ and column $v'$ is 1 if
$(v,v')\in E$ and is 0 otherwise. Conversely, to each $n\times
n$-matrix $P=(p_{ij})$ with non-negative real entries, one can
assign a digraph $D(P)$ on $\{1,2,\dots,n\}$ as the vertex set in
which the pair $(i,j)$ is an edge if and only if $p_{ij}>0$. This
correspondence between matrices and digraphs allows one to state
in the language of digraphs a number of important notions and
results from the classic theory of non-negative matrices (the
Perron--Frobenius theory).

A digraph $D=\langle V,E\rangle$ is said to be \emph{strongly
connected} if for every pair  $(v,v')\in V\times V$, there is a
directed path from $v$ to $v'$. A strongly connected digrpah $D$
is called \emph{primitive} if the greatest common divisor of the
lengths of the directed cycles of $D$ is equal to~1. In the
literature such digraphs are sometimes called \emph{aperiodic}.
Our choice of the name is justified by the fact that a digraph has
this property if and only if its matrix is primitive in the sense
of the Perron--Frobenius theory, that is the matrix has a positive
eigenvalue that is strictly greater than the absolute value of any
of its other elgenvalues.

The $t$-th \emph{power} of a digraph $D=\langle V,E\rangle$ is the digraph $D^t$ with the same vertex set $V$ in which a pair $(v,v')\in
V\times V$ is an edge if and only if $D$ has a directed path from $v$ to $v'$ of length precisely $t$. It is easy to see that if $M$ is the
matrix of $D$, then the digraph $D^t$ is isomorphic to the digraph $D(M^t)$, where $M^t$ is the usual $t$-th power of the matrix $M$. It is
known (see, e.g., \cite[p.\,224]{ST00}) that if $D$ is  a primitive digraph, then for some $t$ the power $D^t$ is a complete digraph (with
loops), that is, in $D^t$ each pair of vertices constitutes an edge. In matrix terms this means that each entry of the matrix $M^t$ is
positive. The least $t$ with this property is called the \emph{exponent} of the digraph $D$ and is denoted by $\gamma(D)$. Exponents of
digraphs have been intensively studied over the last 60 years and we refer to~\cite{Br} for a survey of results accumulated in this area.
In this paper we need only a few classic results collected in the following theorem. For brevity, in this theorem (and in the rest of the
paper) a digraph with $n$ vertices is called an $n$-\emph{digraph}.

\begin{theorem}
\label{dulmage} \emph{(a) (Wielandt's theorem, see
\cite{Wi50,DM62,DM64})} If $D$ is a primitive $n$-digraph, then
$\gamma(D)\le(n-1)^2+1$.

\emph{(b) \cite[Theorem~6 and Corollary 4]{DM64}} If $n>2$, then
up to isomorphism, there is exactly one primitive $n$-digraph $D$
with $\gamma(D)=(n-1)^2+1$, and exactly one with
$\gamma(D)=(n-1)^2$. The matrices of the digraphs are
\begin{equation}
\label{wielandt}
\begin{pmatrix}
0 & 1 & 0 & \dots & 0 & 0\\
0 & 0 & 1 & \dots & 0 & 0\\
\hdotsfor{6}\\
0 & 0 & 0 & \dots & 0 & 1\\
1 & 1 & 0 & \dots & 0 & 0
\end{pmatrix} \text{ and }
\begin{pmatrix}
0 & 1 & 0 & \dots & 0 & 0\\
0 & 0 & 1 & \dots & 0 & 0\\
\hdotsfor{6}\\
1 & 0 & 0 & \dots & 0 & 1\\
1 & 1 & 0 & \dots & 0 & 0
\end{pmatrix}
\text{ respectively.}
\end{equation}

\emph{(c) \cite[Theorem~7]{DM64}} If $n>4$ is even, then there is
no primitive $n$-digraph $D$ such that
$n^2-4n+6<\gamma(D)<(n-1)^2$, and, up to isomorphism, there are
either $3$ or $4$ primitive $n$-digraphs $D$ with
$\gamma(D)=n^2-4n+6$, according as $n$ is or is not a~multiple of
$3$.

\emph{(d) \cite[Theorem~8]{DM64}} If $n>3$ is odd, then there is
no primitive $n$-digraph $D$ such that
$n^2-3n+4<\gamma(D)<(n-1)^2$, and, up to isomorphism, there is
exactly one primitive $n$-digraph $D$ with $\gamma(D)=n^2-3n+4$,
exactly one with $\gamma(D)=n^2-3n+3$, and exactly two with
$\gamma(D)=n^2-3n+2$. The matrices of these digraphs are:
\begin{gather}
\label{odd island}
\begin{pmatrix}
0 & 1 & 0 & \dots & 0 & 0\\
0 & 0 & 1 & \dots & 0 & 0\\
\hdotsfor{6}\\
0 & 0 & 0 & \dots & 1 & 0\\
0 & 0 & 0 & \dots & 0 & 1\\
1 & 0 & 1 & \dots & 0 & 0
\end{pmatrix}\!, \quad
\begin{pmatrix}
0 & 1 & 0 & \dots & 0 & 0\\
0 & 0 & 1 & \dots & 0 & 0\\
\hdotsfor{6}\\
0 & 0 & 0 & \dots & 1 & 0\\
0 & 1 & 0 & \dots & 0 & 1\\
1 & 0 & 1 & \dots & 0 & 0
\end{pmatrix}\!,\\
\label{odd island1}
\begin{pmatrix}
0 & 1 & 0 & \dots & 0 & 0\\
0 & 0 & 1 & \dots & 0 & 0\\
\hdotsfor{6}\\
1 & 0 & 0 & \dots & 1 & 0\\
0 & 1 & 0 & \dots & 0 & 1\\
1 & 0 & 1 & \dots & 0 & 0
\end{pmatrix}\!,\quad
\begin{pmatrix}
0 & 1 & 0 & \dots & 0 & 0\\
0 & 0 & 1 & \dots & 0 & 0\\
\hdotsfor{6}\\
1 & 0 & 0 & \dots & 1 & 0\\
0 & 0 & 0 & \dots & 0 & 1\\
1 & 0 & 1 & \dots & 0 & 0
\end{pmatrix}\!.
\end{gather}

\emph{(e) \cite[Theorem~8]{DM64}} If $n>3$ is odd, then there is
no primitive $n$-digraph $D$ such that
$n^2-4n+6<\gamma(D)<n^2-3n+2$, and, up to isomorphism, there are
either $3$ or $4$ primitive $n$-digraphs $D$ with
$\gamma(D)=n^2-4n+6$, according as $n$ is or is not a~multiple of
$3$.
\end{theorem}

In Table~\ref{9 states}, we compare the experimental data from
Table~\ref{automata} with the data that we can extract from
Theorem~\ref{dulmage}. Both digraphs and automata are counted up
to isomorphism.

\begin{table}[ht]
\tabcolsep=4.2pt \extrarowheight=1pt \caption{Exponents of
primitive 9-digraphs with $9$ vertices vs \rl{}s for 2-letter \sa\
with $9$ states} \label{9 states} {\small
\begin{tabular}{|p{2.9cm}||c|c|c|c|c|c|c|c|c|c|c|c|c|c|c|} \hline
\centering{$N$} & 65 & 64 & 63 & 62 & 61 & 60 & 59 & 58 & 57 & 56 & 55 & 54 & 53 & 52 & 51 \\
\hline \raggedright{\# of digraphs with exponent $N$} &
\raisebox{-6pt}{1} & \raisebox{-6pt}{1} & \raisebox{-6pt}{0} &
\raisebox{-6pt}{0} & \raisebox{-6pt}{0} & \raisebox{-6pt}{0} &
\raisebox{-6pt}{0} & \raisebox{-6pt}{1} & \raisebox{-6pt}{1} &
\raisebox{-6pt}{2}
& \raisebox{-6pt}{0} & \raisebox{-6pt}{0} & \raisebox{-6pt}{0} & \raisebox{-6pt}{0} & \raisebox{-6pt}{3} \\
\hline \raggedright{\# of automata with \rl\ $N$}
&\raisebox{-11pt}{0} &\raisebox{-11pt}{1} &\raisebox{-11pt}{0}
&\raisebox{-11pt}{0} &\raisebox{-11pt}{0} &\raisebox{-11pt}{0}
&\raisebox{-11pt}{0} &\raisebox{-11pt}{1} &\raisebox{-11pt}{2}
&\raisebox{-11pt}{3}
&\raisebox{-11pt}{0} &\raisebox{-11pt}{0} &\raisebox{-11pt}{0} &\raisebox{-11pt}{4} &\raisebox{-11pt}{4} \\
\hline
\end{tabular}}
\end{table}

There is an obvious similarity between the second and the third
rows of Table~\ref{9 states}. An analogous similarity is revealed
when one compares the data for other sizes of automata/digraphs.
We believe that the observed similarity is more than a mere
coincidence and that, in contrary, it reflects some profound and
perhaps yet hidden interconnections between primitive digraphs and
\sa. Some of such interconnections have been discovered in the
course of investigations related to the so-called Road Coloring
Problem. We recall notions involved there.

Given a DFA $\mathrsfs{A}=\langle Q,\Sigma,\delta\rangle$, its
\emph{digraph} $D(\mathrsfs{A})$ has $Q$ as the vertex set and
$(q,q')\in Q\times Q$ is an edge of $D(\mathrsfs{A})$ if and only
if $q'=\delta(q,a)$ for some $a\in\Sigma$. It is easy to see that
a digraph $D$ is isomorphic to the digraph of some DFA if and only
if each vertex of $D$ has at least one outgoing edge. In the
sequel, we always consider only digraphs satisfying this property.
Every DFA $\mathrsfs{A}$ such that $D\cong D(\mathrsfs{A})$ is
called a \emph{coloring} of $D$. Thus, every coloring of $D$ is
defined by assigning non-empty sets of labels (colors) from some
alphabet $\Sigma$ to edges of $D$ such that the label sets
assigned to the outgoing edges of each vertex form a partition of
$\Sigma$. Fig.\,\ref{fig:cerny} shows a digraph and two of its
colorings by $\Sigma=\{a,b\}$.

\begin{figure}[ht]
 \begin{center}
  \unitlength=2.8pt
    \begin{picture}(18,30)(-65,-4)
    \gasset{Nw=6,Nh=6,Nmr=3}
    \node(A1)(0,18){$1$}
    \node(A2)(18,18){$2$}
    \node(A3)(18,0){$3$}
    \node(A4)(0,0){$4$}
    \drawloop[loopangle=135](A1){$a$}
    \drawloop[loopangle=45](A2){$b$}
    \drawloop[loopangle=-45](A3){$b$}
    \drawedge(A1,A2){$b$}
    \drawedge(A2,A3){$a$}
    \drawedge(A3,A4){$a$}
    \drawedge(A4,A1){$a,b$}
    \end{picture}
 \begin{picture}(18,30)(0,-4)
    \gasset{Nw=6,Nh=6,Nmr=3}
    \node(A1)(0,18){$1$}
    \node(A2)(18,18){$2$}
    \node(A3)(18,0){$3$}
    \node(A4)(0,0){$4$}
    \drawloop[loopangle=135](A1){$a$}
    \drawloop[loopangle=45](A2){$a$}
    \drawloop[loopangle=-45](A3){$a$}
    \drawedge(A1,A2){$b$}
    \drawedge(A2,A3){$b$}
    \drawedge(A3,A4){$b$}
    \drawedge(A4,A1){$a,b$}
    \end{picture}
 \begin{picture}(18,30)(65,-4)
    \gasset{Nw=6,Nh=6,Nmr=3}
    \node(A1)(0,18){$1$}
    \node(A2)(18,18){$2$}
    \node(A3)(18,0){$3$}
    \node(A4)(0,0){$4$}
    \drawloop[loopangle=135](A1){}
    \drawloop[loopangle=45](A2){}
    \drawloop[loopangle=-45](A3){}
    \drawedge(A1,A2){}
    \drawedge(A2,A3){}
    \drawedge(A3,A4){}
    \drawedge(A4,A1){}
    \end{picture}
 \end{center}
 \caption{A digraph and two of its colorings}
 \label{fig:cerny}
\end{figure}
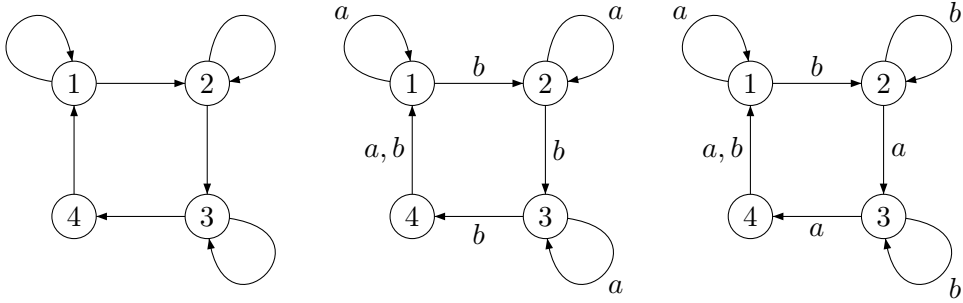

In 1977 Adler, Goodwyn, and Weiss~\cite{AGW} observed that the
digraph of every \scn\ \san\ is primitive and conjectured that
every primitive digraph has a synchronizing coloring. This
conjecture, known as the Road Coloring Conjecture,  has been
recently proved by Trahtman~\cite{Tr09}. It is easy to relate the
\rl\ of a \scn\ \san\ with the exponent of its digraph.

\begin{prop}
\label{lower bound} Let $\mathrsfs{A}=\langle
Q,\Sigma,\delta\rangle$ be a \scn\ synchronizing $n$-automaton
with \rl\ $t$. Then
\begin{equation}
\label{eq:lower bound} \gamma(D(\mathrsfs{A}))\le t+n-1.
\end{equation}
\end{prop}

\begin{proof}
Let $w\in\Sigma^*$ be a \sw\ of length $t$ and let $p$ be such that $\delta(q,w)=p$ for each $q\in Q$. Now we take an arbitrary pair
$(q',q'')\in Q\times Q$ and construct a directed path from $q'$ to $q''$ of length precisely  $t+n-1$ in the digraph $D(\mathrsfs{A})$.
Since the digraph $D(\mathrsfs{A})$ is strongly connected, it has a directed path from $p$ to $q''$. Let $\ell$ be the minimum length of
such a path. Since the path of minimum length visits no vertex twice, $\ell\le n-1$. Now we consider an arbitrary directed path of length
$n-1-\ell$ starting with $q'$. From the endpoint of this path, we walk along the directed path of length $t$ labelled by the word $w$ in
the automaton $\mathrsfs{A}$. Since $w$ is a reset word, this path necessarily ends with $p$. It remains to walk along the directed path of
length $\ell$ from $p$ to $q''$ in order to obtain a directed path from $q'$ to $q''$ of length $(n-1-\ell)+t+\ell=t+n-1$.
\end{proof}

Thus, colorings of primitive digraphs with large exponents yield
automata with large \rl{}s. This observation is not yet sufficient
to completely explain the similarity between the rows of
Table~\ref{9 states} since many automata corresponding to non-zero
entries in the third row cannot be obtained as colorings of
primitive digraphs with large exponents. In Section~4 we present
yet another way to produce slowly \sa\ from primitive digraphs.

\section*{{\centerline{\large\bf \S4. Series of slowly \sa}}}

\paragraph*{4.1. Overview.} First of all, we would like to discuss what kind
of automata we are interested in, in other words, what is the
precise meaning of the expressions like ``slowly \sa'' or
``automata whose \rl\ is close to the square of the state
number''. As mentioned in Section~1, the \rl\ of an $n$-automaton
can reach $(n-1)^2$. The corresponding example (discovered by
\v{C}ern\'y~\cite{Ce64}) is the automaton $\mathrsfs{C}_n=\langle
\{1,2,\dots,n\},\{a,b\},\delta\rangle$ where the letters $a$ and
$b$ act as follows:
$$\delta(i,a)=\begin{cases}
i &\text{if } i<n,\\
1 &\text{if } i=n;
\end{cases}\quad
\delta(i,b)=\begin{cases}
i+1 &\text{if } i<n,\\
1 &\text{if } i=n.
\end{cases}$$
The automaton $\mathrsfs{C}_n$ is shown in Fig.\,\ref{fig:cerny-n}(left). It is easy to see that if one adds a new state $q_0$ to the
automaton $\mathrsfs{C}_{n-1}$ and then defines the action of the letters $a$ and $b$ at this added state in all possible ways such that at
least one of the letters does not fix $q_0$, then one gets $n^2-1$ non-isomorphic initially connected $n$-automata with \rl{}s between
$(n-2)^2$ and $(n-2)^2+1$. In a similar way one can ``multiply'' other $(n-1)$-automata whose \rl\ is close to $(n-2)^2$, thus obtaining
families of $n$-automata.

Since we want to avoid considering such more or less trivial
modifications, we focus on $n$-automata whose \rl{}s are between
$(n-2)^2+2$ and $(n-1)^2$. (This explains, in particular, our
choice of the range of \rl{}s in Tables~\ref{automata} and~\ref{9
states}.)

Our experiments show that the number of \sa\ in this range is not large and that their distribution with respect to the possible values of
\rl\ clearly reveals the following pattern: an isolated extreme value---a gap---a small ``island''---another gap
---a ``continent'', see Table~\ref{automata} and the discussion at
the end of Section~2. We shall show that departing from primitive
digraph with large exponent presented in Theorem~\ref{dulmage},
the following series of automata over 2-letter alphabet can be
constructed:
\begin{itemize}
\item the series $\mathrsfs{C}_n$ that corresponds to the observed extreme value;
\item the series $\mathrsfs{W}_n$, $\mathrsfs{D}'_n$, $\mathrsfs{D}''_n$,
$\mathrsfs{E}_n$, and for odd $n$ also the series $\mathrsfs{B}_n$
and $\mathrsfs{F}_n$ that correspond to all observed ``island''
values;
\item the series $\mathrsfs{G}_n$ (for odd $n$) and $\mathrsfs{H}_n$ that
correspond to the maximum observed ``continental'' values.
\end{itemize}
Fig.~\ref{fig:relation} demonstrates which series of digraphs give
rise to the series of automata just listed. (For clarity, digraph
series in Fig.~\ref{fig:relation} are presented by their ``icons''
rather than matrices from Theorem~\ref{dulmage}.) A solid arrow
from an icon to a letter denoting a series of automata means that
automata in the series are colorings of the corresponding
digraphs; a dotted arrow  indicates another way of producing
automata from digraphs. Dotted frames embrace series of automata
with the same value of \rl.
\begin{figure}[h]
\begin{center}
\includegraphics[width=12cm]{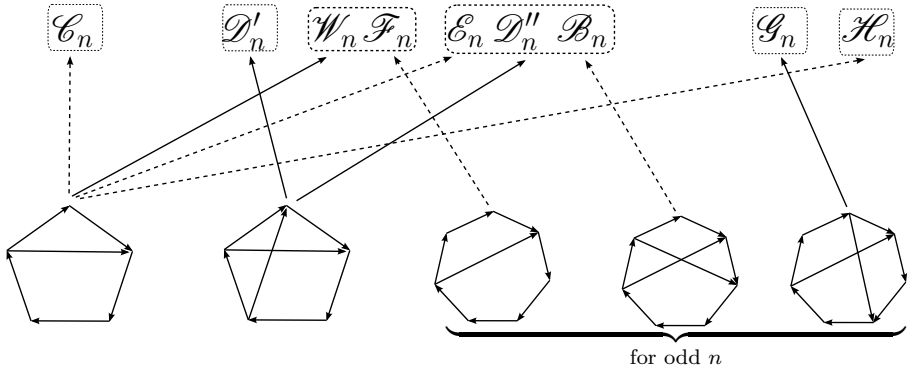}
\unitlength 0.86mm
\begin{picture}(1,1)(2,0)
\put(-132,45){\Large{$\mathrsfs{C}_n$}}
\put(-106,45){\Large{$\mathrsfs{D}'_n$}}
\put(-92,45){\Large{$\mathrsfs{W}_n$}}
\put(-84,45){\Large{$\mathrsfs{F}_n$}}
\put(-71,45){\Large{$\mathrsfs{E}_n$}}
\put(-64,45){\Large{$\mathrsfs{D}''_n$}}
\put(-54,45){\Large{$\mathrsfs{B}_n$}}
\put(-23,45){\Large{$\mathrsfs{G}_n$}}
\put(-10,45){\Large{$\mathrsfs{H}_n$}}
\put(-71,1){$\underbrace{\rule{6.1cm}{0pt}}_{\text{for odd $n$}}$}
\end{picture}
\caption{Connections between series of primitive digraphs with
large exponents and series of slowly \sa} \label{fig:relation}
\end{center}
\end{figure}

Thus, we are going to establish a number of results about automata
series listed in Fig.~\ref{fig:relation}. (Two of these
results---namely, the two concerning the series $\mathrsfs{C}_n$
and $\mathrsfs{B}_n$---are known in the literature but our proofs
are essentially novel.) The results are divided into 5 groups in
the accordance with the ``origin'' of the series, that is,
according to the type of digraphs that give rise to automata in
the series.

\paragraph*{4.2. Automata related to digraphs of the series $W_n$.} The digraph
$W_n$ is the $n$-digraph with largest exponent that corresponds to
the first matrix in~\eqref{wielandt}. If we denote the vertices of
$W_n$ by $1,2,\dots,n$, then its edges are $(n,1)$, $(n,2)$ and
$(i,i+1)$ for $i=1,\dots,n-1$. It is easy to see that up to
isomorphism and renaming of letters, there is a unique coloring of
the digraph $W_n$ with two letters. We denote the resulting
automaton by $\mathrsfs{W}_n$. The digraph $W_n$ and the automaton
$\mathrsfs{W}_n$ are shown in Fig.\,\ref{fig:anan}.

\begin{figure}[thb]
\begin{center}
\unitlength .5mm
\begin{picture}(72,60)(20,-67)
\gasset{Nw=16,Nh=16,Nmr=8} \node(n0)(36.0,-16.0){1} \node(n1)(4.0,-40.0){$n$}
\node(n2)(68.0,-40.0){2} \node(n3)(16.0,-72.0){$n{-}1$}
\node(n4)(56.0,-72.0){3} \drawedge[ELdist=2.0](n1,n0){}
\drawedge[ELdist=1.5](n2,n4){} \drawedge[ELdist=1.7](n0,n2){}
\drawedge[ELdist=1.7](n3,n1){} \drawedge[ELdist=2.0](n1,n2){}
\put(31,-73){$\dots$}
\end{picture}
\begin{picture}(72,60)(-20,-67)
\gasset{Nw=16,Nh=16,Nmr=8}
\node(n0)(36.0,-16.0){1}
\node(n1)(4.0,-40.0){$n$} \node(n2)(68.0,-40.0){2}
\node(n3)(16.0,-72.0){$n{-}1$} \node(n4)(56.0,-72.0){3}
\drawedge[ELdist=2.0](n1,n0){$b$} \drawedge[ELdist=1.5](n2,n4){$a, b$}
\drawedge[ELdist=1.7](n0,n2){$a, b$}
\drawedge[ELdist=1.7](n3,n1){$a, b$} \drawedge[ELdist=2.0](n1,n2){$a$}
\put(31,-73){$\dots$}
\end{picture}
\end{center}
\caption{The digraph $W_n$ and the automaton
$\mathrsfs{W}_n$}\label{fig:anan}
\end{figure}
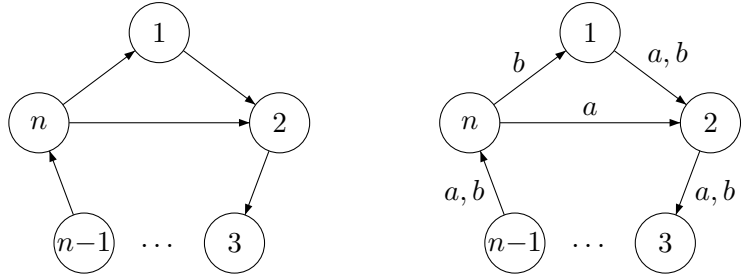

\begin{theorem}
\label{theorem:anan} The automaton $\mathrsfs{W}_n$ is
synchronizing and its \rl\ is equal to $n^2-3n+3$.
\end{theorem}

\begin{proof}
It is easy to see that the word $(ab^{n-2})^{n-2}a$ resets the
automaton $\mathrsfs{W}_n$. The length of this word is equal to
$(n-1)(n-2)+1=n^2-3n+3$. On the other hand,
Theorem~\ref{dulmage}(b) and Proposition~\ref{lower bound} imply
that the \rl\ of $\mathrsfs{W}_n$ cannot be less than меньше чем
$((n-1)^2+1)-(n-1)=n^2-3n+3$.
\end{proof}

The series $\mathrsfs{W}_n$ was discovered by the first author in~2008. His original proof of Theorem~\ref{theorem:anan} relied on a
game-theoretic technique from~\cite{AVZ} and was rather difficult.

Now we show that also the automata in the series $\mathrsfs{C}_n$
are tightly related to the digraphs in the series $W_n$ though the
relation is less obvious. We notice that even though the automata
$\mathrsfs{C}_n$ have been known for about 50 years and have been
rediscovered several times, to the best of our knowledge, their
relationship to the digraphs in the series $W_n$ has not been
observed previously.
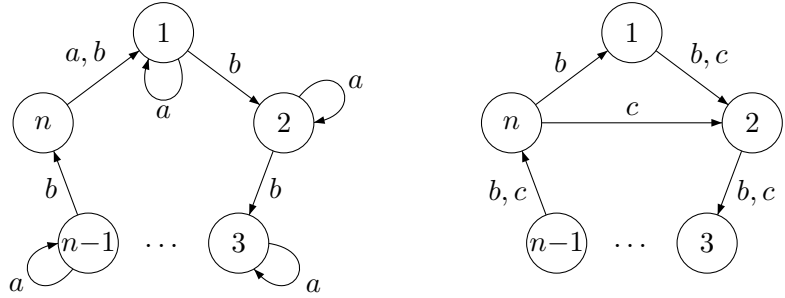
\begin{figure}[bth]
\begin{center}
\unitlength .5mm
\begin{picture}(72,66)(25,-76)
\gasset{Nw=16,Nh=16,Nmr=8,loopdiam=10} \node(n0)(36.0,-16.0){1}
\node(n1)(4.0,-40.0){$n$} \node(n2)(68.0,-40.0){2}
\node(n3)(16.0,-72.0){$n{-}1$} \node(n4)(56.0,-72.0){3}
\drawedge[ELdist=2.0](n1,n0){$a,b$} \drawedge[ELdist=1.5](n2,n4){$b$}
\drawedge[ELdist=1.7](n0,n2){$b$} \drawedge[ELdist=1.7](n3,n1){$b$}
\drawloop[ELdist=1.5,loopangle=30](n2){$a$}
\drawloop[ELdist=2.4,loopangle=-30](n4){$a$}
\drawloop[ELdist=1.5,loopangle=-90](n0){$a$}
\drawloop[ELdist=1.5,loopangle=210](n3){$a$} \put(31,-73){$\dots$}
\end{picture}
\begin{picture}(72,66)(-25,-76)
\gasset{Nw=16,Nh=16,Nmr=8} \node(n0)(36.0,-16.0){1} \node(n1)(4.0,-40.0){$n$}
\node(n2)(68.0,-40.0){2} \node(n3)(16.0,-72.0){$n{-}1$}
\node(n4)(56.0,-72.0){3} \drawedge[ELdist=2.0](n1,n0){$b$}
\drawedge[ELdist=1.5](n2,n4){$b,c$} \drawedge[ELdist=1.7](n0,n2){$b,c$}
\drawedge[ELdist=1.7](n3,n1){$b,c$} \drawedge[ELdist=2.0](n1,n2){$c$}
\put(31,-73){$\dots$}
\end{picture}
\end{center}
\caption{The automaton $\mathrsfs{C}_n$ and the automaton defined
by the action of the words $b$ and $c=ab$}\label{fig:cerny-n}
\end{figure}

We give a new simple proof of \v{C}ern\'y's classical result.
\begin{theorem}[\!\!{\mdseries\cite[Lemma~1]{Ce64}}]
\label{theorem:cerny} The automaton $\mathrsfs{C}_n$ is
synchronizing and its \rl\ is equal to $(n-1)^2$.
\end{theorem}

\begin{proof}
It is easy to see that the word $(ab^{n-1})^{n-2}a$ resets the
automaton $\mathrsfs{C}_n$. The length of this word is equal to
$n(n-2)+1=(n-1)^2$.

Now we invoke the following observation that will be used in some
other proofs as well.

\begin{prop}
\label{prop:simple idempotent} Let $\mathrsfs{A}=\langle
Q,\{a,b\},\delta\rangle$ be a synchronizing $n$-automaton with
\rl\ $t$ in which the letter $a$ fixes all but one states and the
letter $b$ acts as a permutation of the set $Q$. Consider the
automaton $\mathrsfs{B}=\langle Q,\{b,c\},\zeta\rangle$ in which
$\zeta(q,b)=\delta(q,b)$ and $\zeta(q,c)=\delta(q,ab)$ for all
$q\in Q$. Then the automaton $\mathrsfs{B}$ is synchronizing and
its \rl\ does not exceed $t-n+2$.
\end{prop}

\begin{proof}
Let $w$ be a \sw\ of the automaton $\mathrsfs{A}$ of length $t$. Since the letter $b$ acts as a permutation of the set $Q$, the word $w$
cannot end with $b$ for otherwise we could obtain a shorter \sw\ by removing the last letter of $w$. Thus, $w=w'a$ for some word
$w'\in\{a,b\}^*$. Let $q_1\in Q$ be the unique state which is not fixed by the letter $a$ and let $q_2=\delta(q_1,a)$. The minimality of
the length of the word $w$ implies that the image of the set $Q$ under the action of the word $w'$ is equal to $\{q_1,q_2\}$.

Since the word $a^2$ acts on $Q$ in the same way as the letter $a$, this word cannot occur in $w$ as a factor for otherwise we obtain a
shorter \sw\ by substituting the occurrence of $a^2$ in $w$ by $a$. Therefore each occurrence of $a$ in $w$, except the last one, is
followed by an occurrence of the letter $b$. Hence the word $w'$ can be written as a word in the generators $b$ and $ab$. Now we substitute
each occurrence of the factor $ab$ in $w'$ by an occurrence of the letter $c$ so that we rewrite $w'$ into a word $v$ over the alphabet
$\{b,c\}$. Since the words $w'$ and $v$ act on the set $Q$ in the same way, $vc$ is a \sw\ for the automaton $\mathrsfs{B}$. Thus,
$\mathrsfs{B}$ is a \san; let $s$ be its \rl.

Since $b$ only permutes the states and each application of $c$ can
send to one state only one pair of states, the word $vc$ that
sends all states to a single state must contain at least $n-1$
occurrences of $c$. The length of $v$ as a word over $\{b,c\}$ is
not less than  $s-1$ and $v$ contains at least $n-2$ occurrences
of $c$. Each occurrence of $c$ in $v$ corresponds to an occurrence
of the factor $ab$ in $w'$, whence we conclude that the word $w'$
has length at least $(s-1)+(n-2)$. Since the length of the word
$w=w'a$ is equal to $t$, we obtain $t-1\ge (s-1)+(n-2)$, whence
$s\le t-n+2$.
\end{proof}

Observe that in the sequel we will often use modifications of a given automaton $\mathrsfs{A}$ in the flavor of
Proposition~\ref{prop:simple idempotent}. In such modifications, we consider a new automaton on the same state set but with input letters
$c_1$ and $c_2$ whose actions are defined by the actions of some words $w_1$ and $w_2$ respectively in the automaton $\mathrsfs{A}$.
Slightly abusing terminology, we refer to the automaton arising this way as the automaton \emph{defined by the actions of the words
$c_1=w_1$ and $c_2=w_2$}.

We return to the proof of Theorem~\ref{theorem:cerny}. It is easy to see that for the  automaton $\mathrsfs{C}_n$, the automaton defined by
the actions of the words $b$ and $c=ab$ is isomorphic to the automaton $\mathrsfs{W}_n$, see Fig.\,\ref{fig:cerny-n}(right). By
Theorem~\ref{theorem:anan}, the \rl\ of $\mathrsfs{W}_n$ is $n^2-3n+3$. Applying Proposition~\ref{prop:simple idempotent}, we conclude that
the \rl\ of $\mathrsfs{C}_n$ cannot be less than $(n^2-3n+3)+(n-2)=n^2-2n+1=(n-1)^2$.
\end{proof}

The next series in the family of automata related to the digraph $W_n$ consists of the automata
$\mathrsfs{E}_n=\langle\{1,2,\dots,n\},\{a,b\},\delta\rangle$, where the letter $a$ and $b$ act as follows:
$$\delta(i,a)=\begin{cases}
2 &\text{if } i = 1,\\
3 &\text{if } i = 2,\\
i &\text{if } i>2;
\end{cases}\quad
\delta(i,b)=\begin{cases}
i+1 &\text{if } i<n,\\
1 &\text{if } i=n.
\end{cases}$$
The automaton $\mathrsfs{E}_n$ is shown in Fig.\,\ref{fig:e-n}
(left).
\begin{figure}[thb]
\begin{center}
\unitlength .45mm
\begin{picture}(72,66)(25,-76)
\gasset{Nw=16,Nh=16,Nmr=8,loopdiam=10} \node(n0)(36.0,-16.0){2}
\node(n1)(4.0,-40.0){$1$} \node(n2)(68.0,-40.0){3} \node(n3)(16.0,-72.0){$n$}
\node(n4)(56.0,-72.0){4} \drawedge[ELdist=2.0](n1,n0){$a$}
\drawedge[ELdist=1.5](n2,n4){$b$} \drawedge[ELdist=1.7](n0,n2){$a,b$}
\drawedge[ELdist=1.7](n3,n1){$b$} \drawedge[ELdist=1.7](n1,n2){$b$}
\drawloop[ELdist=1.5,loopangle=30](n2){$a$}
\drawloop[ELdist=2.4,loopangle=-30](n4){$a$}
\drawloop[ELdist=1.5,loopangle=210](n3){$a$} \put(31,-73){$\dots$}
\end{picture}
\begin{picture}(72,66)(-25,-76)
\gasset{Nw=16,Nh=16,Nmr=8,loopdiam=10} \node(n0)(36.0,-16.0){2}
\node(n1)(4.0,-40.0){1} \node(n2)(68.0,-40.0){3} \node(n3)(16.0,-72.0){$n$}
\node(n4)(56.0,-72.0){4} \drawedge[ELdist=1.5](n2,n4){$b$}
\drawedge[ELdist=1.7](n0,n2){$b,c$} \drawedge[ELdist=1.7](n3,n1){$b$}
\drawedge[ELdist=2.0](n1,n2){$b,c$} \drawloop[ELdist=1.5,loopangle=30](n2){$c$}
\drawloop[ELdist=2.4,loopangle=-30](n4){$c$}
\drawloop[ELdist=1.5,loopangle=210](n3){$c$} \put(31,-73){$\dots$}
\end{picture}
\end{center}
\caption{The automaton $\mathrsfs{E}_n$ and the automaton defined
by the actions of the words $b$ and $c=aa$}\label{fig:e-n}
\end{figure}
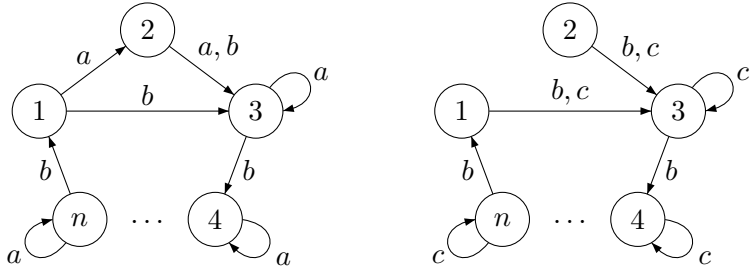

\begin{theorem}
\label{thm:e-n} The automaton  $\mathrsfs{E}_n$ is synchronizing,
and its \rl\ is equal to $n^2-3n+2$.
\end{theorem}

\begin{proof}
It is easy to verify that $(a^2b^{n-2})^{n-3}a^2$ is a \sw\ for
the automaton $\mathrsfs{E}_n$. The length of this word is equal
to $n(n-3)+2=n^2-3n+2$.

Now let $w$ be a \ssw\ for $\mathrsfs{E}_n$. We notice that in
$\mathrsfs{E}_n$ the words $bab$ and $b^2$ act in the same way and
so do the words $a^3$ and $a^2$. Therefore neither $bab$ nor $a^3$
can occur in the word $w$ as a factor. Besides that, $w$ cannot
start with $ab$. Indeed, the image of the set of all states under
the action of the word $ab$ is equal to $\{1,3,\dots,n\}$ and thus
coincides with the image of the letter $b$. Therefore would the
word $w$ start with $ab$, we could obtain a shorter \sw\ by
substituting $ab$ by $b$. Finally, $w$ cannot end with $ba$.
Indeed, if $w=w'a$, then the minimality of $w$ implies that the
image of the set of all states under the action of the word $w'$
is equal to $\{2,3\}$. This set, however, is not contained in the
image of the letter $b$, whence $w'$ cannot end with $b$. Thus,
every occurrence of the letter $a$ in the word $w$ happens within
the factor $a^2$ and no occurrences of these factors in $w$ can
overlap.

Let $c=a^2$, then the word $w$ can be rewritten into a word $v$
over the alphabet $\{b,c\}$. The actions of $b$ and $c$ on the set
$\{1,2,\dots,n\}$ define an automaton shown in Fig.\,\ref{fig:e-n}
(right). Since the words $w$ and $v$ act on $\{1,2,\dots,n\}$ in
the same way, $v$ is a \sw\ for this automaton, and hence, for its
subautomaton on the set $\{1,3,\dots,n\}$. It is easy to see that
the latter subautomaton is isomorphic to the automaton
$\mathrsfs{C}_{n-1}$. By Theorem~\ref{theorem:cerny} the length of
$v$ as a word over $\{b,c\}$ is at least $(n-2)^2$ and $v$
contains at least $n-2$ occurrences of $c$. Since every occurrence
of $c$ in $v$ corresponds to an occurrence of the factor $a^2$ in
$w$, we conclude that the length of word $w$ is not less than
$(n-2)^2+(n-2)=n^2-3n+2$.
\end{proof}

The proof of Theorem~\ref{thm:e-n} shows that the automaton
$\mathrsfs{E}_n$ arises from one of the ``trivial'' modifications
of the automaton $\mathrsfs{C}_{n-1}$ that we discussed in
Subsection~4.1. The last series of slowly \sa\ from the automata
family related to the digraph $W_n$ arises from a similar
modification of the automaton $\mathrsfs{W}_{n-1}$. The series
consists of the automata
$\mathrsfs{H}_n=\langle\{1,2,\dots,n\},\{a,b\},\delta\rangle$,
where the letter $a$ and $b$ act as follows:
$$\delta(i,a)=\begin{cases}
n &\text{if } i=1,\\
i &\text{if } 1<i<n,\\
1 &\text{if } i=n;
\end{cases}\quad
\delta(i,b)=\begin{cases}
i+1 &\text{if } i<n-1,\\
1 &\text{if } i=n-1,\\
3 &\text{if } i=n.
\end{cases}$$
The automaton $\mathrsfs{H}_n$ is shown in Fig.\,\ref{fig:h-n}
(left).

\begin{figure}[ht]
\unitlength .5mm
\begin{center}
\begin{picture}(72,72)(25,-75)
\gasset{Nw=16,Nh=16,Nmr=8,loopdiam=10}
\node[NLangle=0.0](n0)(24.25,-23.6){$n{-}1$}
\node[NLangle=0.0](n1)(52.25,-7.6){$1$}
\node[NLangle=0.0](n2)(80.25,-23.6){$2$}
\node[NLangle=0.0](n3)(24.25,-52){$n{-}2$}
\node[NLangle=0.0](n4)(55,-67.6){$4$} \node[NLangle=0.0](n6)(80.25,-52){$3$}
\drawedge(n1,n2){$b$} \drawedge(n2,n6){$b$} \drawedge(n0,n1){$b$}
\drawedge(n3,n0){$b$} \drawedge(n6,n4){$b$} \drawloop[loopangle=0.0](n6){$a$}
\drawloop[loopangle=270.0](n4){$a$} \drawloop[loopangle=180.0](n3){$a$}
\drawloop[loopangle=0.0](n2){$a$} \node[NLangle=0.0](n31)(52.0,-38){$n$}
\drawedge[ELside=r,ELdist=3.0,curvedepth=-5.0](n1,n31){$a$}
\drawedge[ELside=r,ELdist=3.0,curvedepth=-5.0](n31,n1){$a$}
\drawloop[loopangle=180.0](n0){$a$} \drawedge(n31,n6){$b$}
\put(31,-67){$\dots$}
\end{picture}
\begin{picture}(72,72)(-25,-75)
\gasset{Nw=16,Nh=16,Nmr=8,loopdiam=10}
\node[NLangle=0.0](n0)(24.25,-23.6){$n{-}1$}
\node[NLangle=0.0](n1)(52.25,-7.6){$1$}
\node[NLangle=0.0](n2)(80.25,-23.6){$2$}
\node[NLangle=0.0](n3)(24.25,-52){$n{-}2$}
\node[NLangle=0.0](n4)(55,-67.6){$4$} \node[NLangle=0.0](n6)(80.25,-52){$3$}
\drawedge[ELdist=0.7](n1,n2){$b$} \drawedge[ELpos=35, ELdist=0.5](n1,n6){$c$}
\drawedge(n2,n6){$b, c$} \drawedge(n0,n1){$b, c$} \drawedge(n3,n0){$b, c$}
\drawedge(n6,n4){$b,c$} \node[NLangle=0.0](n31)(50,-40){$n$}
\drawedge[ELpos=35,ELdist=0.7](n31,n6){$b$}
\drawedge[ELpos=35,ELdist=0.7](n31,n2){$c$} \put(31,-67){$\dots$}
\end{picture}
\end{center}
\caption{The automaton $\mathrsfs{H}_n$ and the automaton defined
by the actions of the words $b$ and $c=ab$}\label{fig:h-n}
\end{figure}
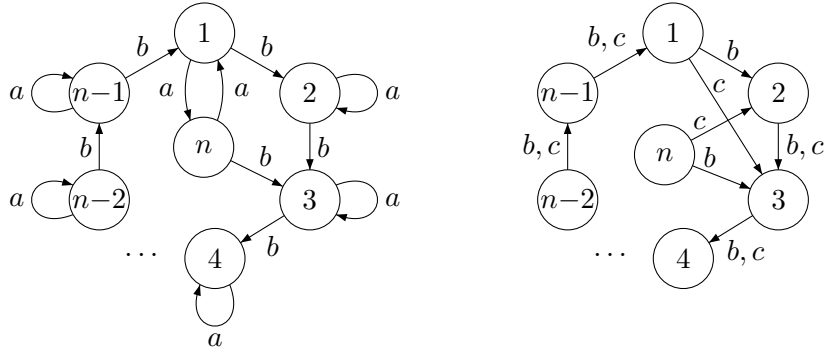

\begin{theorem}
\label{thm:h-n} The automaton $\mathrsfs{H}_n$ is synchronizing
and its \rl\ is equal to $n^2-4n+6$.
\end{theorem}

\begin{proof}
It is easy to check that the words $b(ab^{n-2})^{n-3}ab$ resets
the automaton $\mathrsfs{H}_n$. The length of this word is equal
to $1+(n-1)(n-3)+2=n^2-4n+6$.

Now let $w$ be a \ssw\ for $\mathrsfs{H}_n$. Since the word $a^2$
acts in $\mathrsfs{H}_n$ as the identity transformation, it cannot
occur as a factor in $w$ Besides that, the word $w$ neither starts
nor ends with the letter $a$ because this letter acts as a
permutation of the state set of $\mathrsfs{H}_n$.

Let $c=ab$, then the word $w$ can be rewritten into a word $v$
over the alphabet $\{b,c\}$. The actions of $b$ and $c$ on the set
$\{1,2,\dots,n\}$ define an automaton shown in Fig.\,\ref{fig:h-n}
(right). Since the words $w$ and $v$ act on $\{1,2,\dots,n\}$ in
the same way, $v$ is a \sw\ for this automaton. We have noticed
that the word $w$ starts with the letter $b$, hence so does the
word $v$. If we write $v=bv'$ for some $v'\in\{b,c\}^*$, then it
is easy to see that $v'$ is a \sw\ for the subautomaton on the set
$\{1,2,\dots,n-1\}$. Since this subautomaton is isomorphic to the
automaton $\mathrsfs{W}_{n-1}$, Theorem~\ref{theorem:anan} implies
that the length of $v'$ as a word over $\{b,c\}$ is at least
$(n-1)^2-3(n-1)+3$. Besides that, $v'$ contains at least $n-2$
occurrences of $c$ because $b$ only permutes the states and each
application of $c$ can send only one pair of states to a single
state. Since every occurrence of $c$ in $v'$ corresponds to an
occurrence of the factor $ab$ in $w$, we can conclude that the
length of the word $w$ is not less than
$1+((n-1)^2-3(n-1)+3)+(n-2)=n^2-4n+6$.
\end{proof}

\paragraph*{4.3. Automata related to digraphs of the series $D_n$.} The digraph $D_n$
is the $n$-digraph with exponent $(n-1)^2$ that corresponds to the
second matrix in~\eqref{wielandt}. It can be obtained from the
digraph $W_n$ by adding the edge $(n-1,1)$. It is easy to see that
up to isomorphism and renaming of letters, there exist exactly two
colorings of the digraph $D_n$ with two letters.
Fig.\,\ref{fig:dulmage} shows the digraph $D_n$ and two its
colorings, the automata $\mathrsfs{D}'_n$ and $\mathrsfs{D}''_n$.

\begin{figure}[thb]
\begin{center}
\unitlength .5mm
\begin{picture}(64,60)(30,-72)
\gasset{Nw=16,Nh=16,Nmr=8} \node(n0)(32.0,-16.0){1} \node(n1)(0,-40.0){$n$}
\node(n2)(64.0,-40.0){2} \node(n3)(12.0,-72.0){$n{-}1$}
\node(n4)(52.0,-72.0){3} \drawedge(n1,n0){} \drawedge(n2,n4){}
\drawedge(n0,n2){} \drawedge(n3,n0){} \drawedge(n3,n1){} \drawedge(n1,n2){}
\put(31,-73){$\dots$}
\end{picture}
\begin{picture}(64,60)(0,-72)
\gasset{Nw=16,Nh=16,Nmr=8} \node(n0)(32.0,-16.0){1} \node(n1)(0,-40.0){$n$}
\node(n2)(64.0,-40.0){2} \node(n3)(12.0,-72.0){$n{-}1$}
\node(n4)(52.0,-72.0){3} \drawedge[ELdist=2.0](n1,n0){$b$}
\drawedge[ELdist=1.5](n2,n4){$a, b$} \drawedge[ELdist=1.7](n0,n2){$a,b$}
\drawedge[ELpos=40,ELdist=1.7](n3,n0){$a$} \drawedge[ELdist=1.7](n3,n1){$b$}
\drawedge[ELdist=2.0](n1,n2){$a$} \put(31,-73){$\dots$}
\end{picture}
\begin{picture}(64,60)(-30,-72)
\gasset{Nw=16,Nh=16,Nmr=8} \node(n0)(32.0,-16.0){1} \node(n1)(0,-40.0){$n$}
\node(n2)(64.0,-40.0){2} \node(n3)(12.0,-72.0){$n{-}1$}
\node(n4)(52.0,-72.0){3} \drawedge[ELdist=2.0](n1,n0){$a$}
\drawedge[ELdist=1.5](n2,n4){$a, b$} \drawedge[ELdist=1.7](n0,n2){$a, b$}
\drawedge[ELpos=40,ELdist=1.7](n3,n0){$a$} \drawedge[ELdist=1.7](n3,n1){$b$}
\drawedge[ELdist=2.0](n1,n2){$b$} \put(31,-73){$\dots$}
\end{picture}
\end{center}
\caption{The digraph $D_n$ and its colorings $\mathrsfs{D}'_n$ and
$\mathrsfs{D}''_n$}\label{fig:dulmage}
\end{figure}
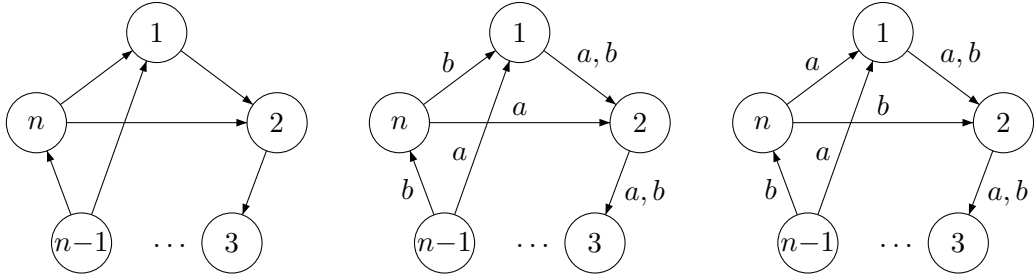

\begin{theorem}
\label{theorem:new series} The automata $\mathrsfs{D}'_n$ and
$\mathrsfs{D}''_n$ are synchronizing and its \rl{}s are equal to
$n^2-3n+4$ and $n^2-3n+2$ respectively.
\end{theorem}

\begin{proof}
It is not hard to verify that the word $(ab^{n-2})^{n-2}ba$ is a \sw\ for the automaton $\mathrsfs{D}'_n$ and the word $(ba^{n-1})^{n-3}ba$
is a \sw\ for the automaton $\mathrsfs{D}''_n$. The lengths of these words are equal $(n-1)(n-2)+2=n^2-3n+4$ and $n(n-3)+2=n^2-3n+2$
respectively.

Theorem~\ref{dulmage}(b) and Proposition~\ref{lower bound} imply
that the \rl\ for the colorings of the digraph $D_n$ cannot be
less than $(n-1)^2-(n-1)=n^2-3n+2$. This proves our theorem for
the automaton $\mathrsfs{D}''_n$.

Now consider the automaton $\mathrsfs{D}'_n$. Here we shall make
use of the following elementary result.
\begin{lemma}[\!\!{\mdseries\cite[Theorem~2.1.1]{RaAl05}}]
\label{sylvester} If $k,\ell$ are relatively prime positive
integers, then $k\ell-k-\ell$ is the largest integer that is not
expressible as a non-negative integer combination of $k$ and
$\ell$.
\end{lemma}

Let $w$ be a \ssw\ for the automaton $\mathrsfs{D}'_n$. Since 2 is the unique state in $\mathrsfs{D}'_n$ that is a common end of two
different edges with the same label, the minimality of $w$ implies that $w$ sends all states of the automaton to~2. Suppose that the length
of $w$ is equal to $n^2-3n+2$. Then the digraph $D_n$ has a directed path of this length from~1 to~2. There is a unique edge starting at~1,
namely, $(1,2)$, hence the path consists of this edge followed by a directed cycle of length $n^2-3n+1$. The digraph $D_n$ has exactly
three simple directed cycles: one of length $n$ and two of length $n-1$. Every directed cycle consists of simple directed cycles whence the
number $n^2-3n+1$ (as the length of a directed cycle in $D_n$) must be  a non-negative integer combination of the numbers $n$ and $n-1$
(the lengths of simple directed cycles). However this is impossible by Lemma~\ref{sylvester} since $n^2-3n+1=n(n-1)-n-(n-1)$.

Now suppose that the length of $w$ is equal to $n^2-3n+3$. Then
the digraph $D_n$ has a directed path of this length from~$n-1$
to~2. Since $b$ acts as a permutation of the state set of the
automaton $\mathrsfs{D}'_n$, the word $w$ starts with the letter
$a$. The state $n-1$ under the action of $a$ goes to the state~1.
Therefore $D_n$ has also a directed path of length $n^2-3n+2$
from~1 to~2 but in the previous paragraph we have shown that this
is impossible. Thus, the length of $w$ cannot be less than
$n^2-3n+4$.
\end{proof}

The series $\mathrsfs{D}'_n$ is of interest because for $n>6$, the automata of this series have the largest  \rl\ among all known automata
except the ones from the \v{C}ern\'y series $\mathrsfs{C}_n$ as well as the largest  \rl\ among all known automata without loops. The
series $\mathrsfs{D}''_n$ also possess an extremal property: the automata from this series have the largest \rl\ among all known automata
in which no letter acts as a permutation of the state set.

There is one further series of slowly \sa\ related to the digraphs $D_n$; it consists of $n$-automata with \rl\ $n^2-4n+6$. We do not
present it here since one series with the same parameters has already been described above, see Theorem~\ref{thm:h-n}.

\paragraph*{4.4. Automata related to digraphs of the series $V_n$.} The digraph $V_n$
is the $n$-digraph corresponding to the first matrix in~\eqref{odd
island}. If we denote the vertices of $V_n$ by $1,2,\dots,n$, then
its edges are суть $(n,1)$, $(n,3)$ and $(i,i+1)$ for
$i=1,\dots,n-1$. The digraph $V_n$ is primitive only when $n$ is
odd, and in this case its exponent is equal to $n^2-3n+4$. The
digraphs of the series $V_n$ give rise to the family of automata
$\mathrsfs{F}_n=\langle\{1,2,\dots,n\},\{a,b\},\delta\rangle$ in
which the letters $a$ and $b$ act as follows:
$$\delta(i,a)=\begin{cases}
i &\text{if } i<n,\\
2 &\text{if } i=n;
\end{cases}\quad
\delta(i,b)=\begin{cases}
i+1 &\text{if } i<n,\\
1 &\text{if } i=n.
\end{cases}$$
The automaton $\mathrsfs{F}_n$ is shown in Fig.\,\ref{fig:f-n}
(left).

\begin{figure}[th]
\begin{center}
\unitlength .5mm
\begin{picture}(72,86)(25,-76)
\gasset{Nw=16,Nh=16,Nmr=8,loopdiam=10} \node(n0)(36.0,-16.0){1}
\node(n1)(4.0,-40.0){$n$} \node(n2)(68.0,-40.0){2}
\node(n3)(16.0,-72.0){$n{-}1$} \node(n4)(56.0,-72.0){3}
\drawedge[ELdist=2.0](n1,n0){$b$} \drawedge[ELdist=1.5](n2,n4){$b$}
\drawedge[ELdist=1.7](n0,n2){$b$} \drawedge[ELdist=1.7](n3,n1){$b$}
\drawedge[ELdist=1.7](n1,n2){$a$} \drawloop[ELdist=1.5,loopangle=30](n2){$a$}
\drawloop[ELdist=2.4,loopangle=-30](n4){$a$}
\drawloop[ELdist=1.5,loopangle=210](n3){$a$}
\drawloop[ELdist=1.5,loopangle=90](n0){$a$} \put(31,-73){$\dots$}
\end{picture}
\begin{picture}(72,86)(-25,-76)
\gasset{Nw=16,Nh=16,Nmr=8} \node(n0)(36.0,-16.0){1} \node(n1)(4.0,-40.0){$n$}
\node(n2)(68.0,-40.0){2} \node(n3)(16.0,-72.0){$n{-1}$}
\node(n4)(56.0,-72.0){3} \drawedge[ELdist=1.5](n2,n4){$b,c$}
\drawedge[ELdist=1.7](n0,n2){$b,c$} \drawedge[ELdist=1.7](n3,n1){$b,c$}
\drawedge[ELdist=2.0](n1,n0){$b$} \drawedge[ELdist=2.0](n1,n4){$c$}
\put(31,-73){$\dots$}
\end{picture}
\end{center}
\caption{The automaton $\mathrsfs{F}_n$ and the automaton, defined
by the actions of the words $b$ and $c=ab$}\label{fig:f-n}
\end{figure}
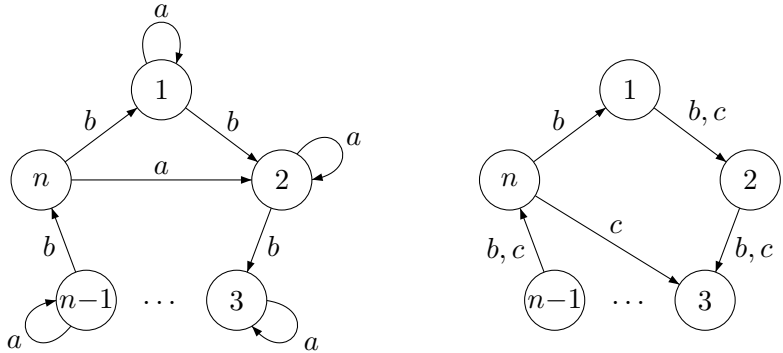

\begin{theorem}
\label{thm:f-n} For odd $n>3$, the automaton $\mathrsfs{F}_n$ is
synchronizing and its \rl\ is equal to $n^2-3n+3$.
\end{theorem}

\begin{proof}
It can be easily verified that for each odd $n>3$ the word
$(ab^{n-2})^{n-2}a$ is a \sw\ for $\mathrsfs{F}_n$. The length of
this word is $(n-1)(n-2)+1=n^2-3n+3$.

Clearly, the automaton $\mathrsfs{F}_n$ satisfies the condition of Proposition~\ref{prop:simple idempotent}. The action of the words $b$
and $c=ab$ on the set $\{1,2,\dots,n\}$ defines an automaton shown in Fig.\,\ref{fig:f-n} (right); we denote this automaton by
$\mathrsfs{V}$. It is easy to see that the automaton $\mathrsfs{V}$ is isomorphic to a coloring of the digraph $V_n$.
Theorem~\ref{dulmage}(d) and Proposition~\ref{lower bound} imply that the \rl\ for colorings of the digraph $V_n$ cannot be less than
$(n^2-3n+4)-(n-1)=n^2-4n+5$. Applying Proposition~\ref{prop:simple idempotent}, we conclude that \rl\ for $\mathrsfs{F}_n$ cannot be less
than $(n^2-4n+5)+(n-2)=n^2-3n+3$.
\end{proof}

\paragraph*{4.5. Automata related to digraphs of the series $R_n$.} The
digraph $R_n$ is the $n$-digraph corresponding to the second matrix in~\eqref{odd island}. One obtains it from the digraph $V_n$ by adding
the edge $(n-1,2)$. The digraph $R_n$ is primitive only when $n$ is odd, and in this case its exponent is equal to $n^2-3n+3$. The digraphs
of the series $R_n$ give rise to the family of automata $\mathrsfs{B}_n=\langle\{1,2,\dots,n\},\{a,b\},\delta\rangle$ in which the letters
$a$ and $b$ act as follows:
$$\delta(i,a)=\begin{cases}
i &\text{if } i<n-1,\\
1 &\text{if } i=n-1,\\
2 &\text{if } i=n;
\end{cases}\quad
\delta(i,b)=\begin{cases}
i+1 &\text{if } i<n,\\
1 &\text{if } i=n.
\end{cases}$$
The automaton $\mathrsfs{B}_n$ is shown in Fig.\,\ref{fig:avz-n}
(left).

\begin{figure}[ht]
\begin{center}
\unitlength .5mm
\begin{picture}(72,72)(30,-76)
\gasset{Nw=16,Nh=16,Nmr=8,loopdiam=10} \node(n1)(56.0,-16.0){1}
\node(n2)(68.0,-44.0){2} \node(n3)(56.0,-72.0){3}
\node(n4)(16.0,-72.0){$n{-}2$} \node(n5)(4.0,-44.0){$n{-}1$}
\node(n6)(16.0,-16.0){$n$} \drawedge[ELdist=1.7](n1,n2){$b$}
\drawedge[ELdist=1.7](n2,n3){$b$} \drawedge[ELdist=1.7](n4,n5){$b$}
\drawedge[ELdist=1.7](n5,n6){$b$} \drawedge[ELdist=1.7](n6,n1){$b$}
\drawedge[ELdist=1.7,ELpos=40](n5,n1){$a$}
\drawedge[ELdist=1.7,ELpos=60](n6,n2){$a$}
\drawloop[ELdist=1.5,loopangle=30](n1){$a$}
\drawloop[ELdist=1.5,loopangle=0](n2){$a$}
\drawloop[ELdist=1.5,loopangle=-30](n3){$a$}
\drawloop[ELdist=1.5,loopangle=210](n4){$a$} \put(31,-73){$\dots$}
\end{picture}
\begin{picture}(72,72)(-30,-76)
\gasset{Nw=16,Nh=16,Nmr=8} \node(n1)(56.0,-16.0){1} \node(n2)(68.0,-44.0){2}
\node(n3)(56.0,-72.0){3} \node(n4)(16.0,-72.0){$n{-}2$}
\node(n5)(4.0,-44.0){$n{-}1$} \node(n6)(16.0,-16.0){$n$}
\drawedge[ELdist=1.7](n1,n2){$b,c$} \drawedge[ELdist=1.7](n2,n3){$b,c$}
\drawedge[ELdist=1.7](n4,n5){$b,c$} \drawedge[ELdist=1.7](n5,n6){$b$}
\drawedge[ELdist=1.7](n6,n1){$b$} \drawedge[ELdist=1.7,ELpos=60](n5,n2){$c$}
\drawedge[ELdist=1.7,ELpos=35](n6,n3){$c$} \put(31,-73){$\dots$}
\end{picture}
\end{center}
\caption{The automaton $\mathrsfs{B}_n$ and the automaton, defined
by the action of the words $b$ and $c=ab$}\label{fig:avz-n}
\end{figure}
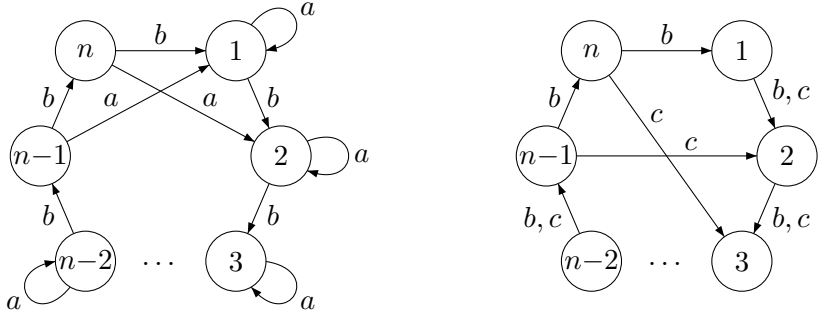

The series $\mathrsfs{B}_n$ (for odd $n>3$) was published
in~\cite{AVZ} and up to recently, it remained the only infinite
series of slowly \sa\ with two input letters in the literature
besides the \v{C}ern\'y series. The fact that the series
$\mathrsfs{B}_n$ is related to the digraphs form the series $R_n$
has not been reported earlier.

The next statement has been the main result of~\cite{AVZ} where it
has been proved by a game-theoretic method. Here we present a
completely elementary proof similar to the proofs of
Theorems~\ref{theorem:cerny} and~\ref{thm:f-n}.

\begin{theorem}[\!\!{\mdseries\cite[Theorem~1.1]{AVZ}}]
\label{theorem:avz} If $n>3$ is odd, the automaton
$\mathrsfs{B}_n$ is synchronizing and its \rl\ is equal to
$n^2-3n+2$.
\end{theorem}

\begin{proof}
For each odd $n>3$ the word
$(ab^{n-2})^{\frac{n-3}{2}}ab^{n-3}(ab^{n-2})^{\frac{n-3}{2}}a$ is
easily seen to be a \sw\ for the automaton $\mathrsfs{B}_n$. The
length of this word is equal to
$(n-1)\frac{n-3}{2}+n-2+(n-1)\frac{n-3}{2}+1=n^2-3n+2.$

Let $w$ be a \ssw\ for the automaton $\mathrsfs{B}_n$ and let $t$
be the length of $w$. Since the letter $b$ acts as a permutation
of the state set, the word $w$ neither starts nor ends with $b$.
In particular, $w=w'a$ for some word $w'\in\{a,b\}^*$. The
minimality of the word $w$ implies that the image of the state set
under the action of the word $w'$ is equal to either $\{1,n-1\}$
or $\{2,n\}$.

Since the word $a^2$ acts in $\mathrsfs{B}_n$ in the same way as
the letter $a$, this word cannot occur in $w$ as a factor.
Further, the word $b^n$ acts in $\mathrsfs{B}_n$ as the identity
transformation and hence it also cannot occur as a factor in a
\ssw. Thus, we conclude that $w=ab^{k_1}ab^{k_2}a\cdots
ab^{k_m}a$, where $1\le k_1,k_2,\dots,k_m\le n-1$.

Let $c=ab$. Then the word $w'$ and the word $v=cb^{k_1-1}cb^{k_2-1}c\cdots cb^{k_m-1}$ act on the set $\{1,2,\dots,n\}$ in the same way.
Therefore the word $vc$ is a \sw\ for the automaton $\mathrsfs{R}$ defined by the actions of the words $b$ and $c=ab$ and shown in
Fig.\,\ref{fig:avz-n} (right). It is clear that the automaton $\mathrsfs{R}$ is isomorphic to a coloring of the digraph $R_n$.
Theorem~\ref{dulmage}(d) and Proposition~\ref{lower bound} imply that the \rl\ for colorings of the digraph $R_n$ cannot be less than
$(n^2-3n+3)-(n-1)=n^2-4n+4$ whence the length of $v$ as a word over $\{b,c\}$, that is, $\sum_{i=1}^mk_i$, is not less than $n^2-4n+3$.
Since $k_i\le n-1$ for all $i=1,\dots,m$, we have
\begin{equation}
\label{inequality} m(n-1)\ge\sum_{i=1}^mk_i\ge n^2-4n+3=(n-3)(n-1),
\end{equation}
whence $m\ge n-3$. The equality $m=n-3$ is only possible when all
inequalities in~\eqref{inequality} become equalities, that is when
$k_i=n-1$ for all $i=1,\dots,m$. In this case
$vc=(cb^{n-2})^{n-3}c$, but this word is not a \sw\ for
$\mathrsfs{R}$ since, as it easy to see, this word permutes the
states~2 and~3. Hence, $m\ge n-2$.

Since every occurrence of $c$ in $v$ corresponds to an occurrence
of the factor $ab$ in $w'$, we conclude that the length of $w'$ is
at least $(n^2-4n+3)+(n-2)=n^2-3n+1$, whence the length of $w$ is
at least $n^2-3n+2$.
\end{proof}

\paragraph*{4.6. Automata related to digraphs of the series $G_n$.}
The digraph $G_n$ is the $n$-digraph corresponding to the second matrix in~\eqref{odd island1}.  One obtains it from the digraph $V_n$ by
adding the edge $(n-2,1)$. The digraph $G_n$ is primitive only when $n$ is odd, and in this case its exponent is equal to $n^2-3n+2$.
Fig.~\ref{fig:aut:fat} shows the automaton $\mathrsfs{G}_n$ which is one of possible colorings of the digraph $G_n$. This series is
interesting for us because for odd $n$, its automata attain the maximal observed ``continental'' value of \rl.

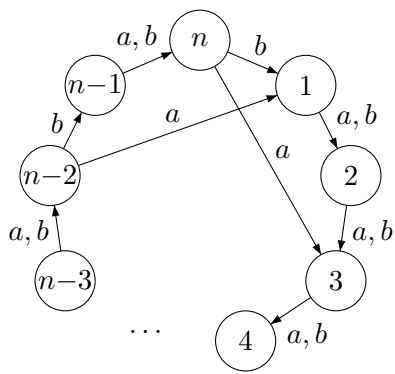
\begin{figure}[ht]
\begin{center}
\unitlength .5mm
\begin{picture}(100,95)(0,-95)
\gasset{Nw=16,Nh=16,Nmr=8} \node[NLangle=0.0](n0)(24.09,-24.01){$n{-}1$}
\node[NLangle=0.0](n1)(51.93,-12.0){$n$}
\node[NLangle=0.0](n2)(80.09,-24.01){$1$}
\node[NLangle=0.0](n3)(12.09,-48.01){$n{-}2$}
\node[NLangle=0.0](n4)(16.09,-76.01){$n{-}3$}
\node[NLangle=0.0](n5)(92.09,-48.01){$2$}
\node[NLangle=0.0](n6)(64.09,-92.01){$4$}
\node[NLangle=0.0](n7)(88.09,-76.01){$3$} \drawedge(n1,n2){$b$}
\drawedge(n2,n5){$a, b$} \drawedge(n5,n7){$a, b$} \drawedge(n7,n6){$a,b$}
\drawedge(n4,n3){$a, b$} \drawedge(n3,n0){$b$} \drawedge(n0,n1){$a,b$}
\drawedge(n1,n7){$a$} \drawedge(n3,n2){$a$} \put(33,-90){$\dots$}
\end{picture}
\end{center}
\caption{The automaton $\mathrsfs{G}_n$} \label{fig:aut:fat}
\end{figure}

\begin{theorem}
For odd $n>3$, the automaton $\mathrsfs{G}_n$ is synchronizing and
its \rl\ is equal to $n^2-4n+7$.
\end{theorem}

\begin{proof}
It is easy to see that for odd $n>3$, the word
$a^2(baba^{n-3})^{n-4}baba^2$ is a \sw\ for the automaton
$\mathrsfs{G}_n$. The length of this word is equal to
$2+n(n-4)+5=n^2-4n+7$.

The following arguments are quite similar to ones from the proof
of Theorem~\ref{theorem:new series}. Theorem~\ref{dulmage}(d) and
Proposition~\ref{lower bound} imply that the \rl\ for colorings of
the digraph $G_n$ cannot be the less than
$(n^2-3n+2)-(n-1)=n^2-4n+3$.

Now let $w$ be a \ssw\ for the automaton $\mathrsfs{G}_n$. Since 3 is the unique state in $\mathrsfs{G}_n$ that is a common end of two
different edges with the same label, the minimality of $w$ implies that $w$ sends all states of the automaton to~3. Suppose that the length
of $w$ is equal to $n^2-4n+3$. Then the digraph $G_n$ has a directed path of this length from~2 to~3. There is a unique edge starting at~2,
namely, $(2,3)$, hence the path consists of this edge followed by a directed cycle of length $n^2-4n+2$. The digraph $G_n$ has exactly
three simple directed cycles: one of length $n$ and two of length $n-2$. Observe that $n$ and $n-2$ are relatively prime since $n$ is odd.
Every directed cycle consists of simple directed cycles whence the number $n^2-4n+2$ (as the length of a directed cycle in $G_n$) must be a
non-negative integer combination of the numbers $n$ and $n-2$ (the lengths of simple directed cycles). However this is impossible by
Lemma~\ref{sylvester} since $n^2-4n+2=n(n-2)-n-(n-2)$.

Suppose that the length of $w$ is equal to $n^2-4n+4$. Then the
digraph $G_n$ has a directed path of this length from~1 to~3.
There is a unique edge starting at~1, namely, $(1,2)$, hence the
path consists of this edge followed by a directed path of length
$n^2-4n+3$ from~2 to~3. In the previous paragraph we have shown
that $G_n$ contains no directed path from~2 to~3 with this length.

Suppose that the length of $w$ is equal to $n^2-4n+5$. The word
$w$ sends $n-2$ to~3. There are two edges starting at $n-2$: есть
the edge $(n-2,1)$ labelled $a$ and the edge $(n-2,n-1)$ labelled
$b$. Since the letter $b$ acts as a permutation on the state set
of the automaton $\mathrsfs{G}_n$, the word $w$ starts with the
letter $a$. Therefore the first edge of the directed path from
$n-2$ to~3 labelled by $w$ is necessarily to edge $(n-2,1)$ and
this edge is followed by a directed path of length $n^2-4n+4$
from~1 to~3. In the previous paragraph we have shown that $G_n$
contains no directed path from~1 to~3 with this length.

Finally, let the length of $w$ is $n^2-4n+6$. The word $w$ sends
each of the states $n-3$ and $n-1$ to the state~3. If the second
letter of the word $w$ is $a$, then the directed path from $n-3$
to~3 labelled by $w$ starts with the edges $(n-3,n-2)$ and
$(n-2,1)$ which are followed by a directed path of length
$n^2-4n+4$ from~1 to~3, and such a path is impossible. If the
second letter of the word $w$ is  $b$, the directed path from
$n-1$ to~3 labelled by $w$ starts with the edges $(n-1,n)$ and
$(n,1)$, which again must be followed by an impossible directed
path of length $n^2-4n+4$ from~1 to~3.

Thus, we have proved that the \rl\ of the automaton
$\mathrsfs{G}_n$ cannot be less than $n^2-4n+7$.
\end{proof}

\section*{{\centerline{\large\bf \S5. Discussion and new conjectures}}}

\paragraph*{5.1. Two conjectures.} The constructions and the results presented in Section~4
witness that the interconnections between \rl{}s of automata with two input letters and exponents of primitive digraphs are sufficiently
tight. This conclusion is also supported by recent results by the third author~\cite{Gusev:2011}. We believe that these interconnections
deserve being further investigated. In order to make the future investigations more concrete, we formulate a very general conjecture in the
flavor of Theorem~\ref{dulmage}. This conjecture constitutes a strengthening of the \v{C}ern\'y conjecture for the case of automata with
two input letters and agrees with all theoretical and experimental results that we are aware of (including the most recent experimental
results from \cite{KoSz13}).

\begin{conjecture}
\label{general} \emph{(a) (The \v{C}ern\'y conjecture)} The \rl\
of every synchronizing $n$-automaton with two input letters does
not exceed $(n-1)^2$.

\emph{(b)} If $n>6$, then up to isomorphism there exists exactly
one synchronizing $n$-automaton with two input letters and \rl\
$(n-1)^2$, namely, the automaton $\mathrsfs{C}_n$.

\emph{(c)} If $n>6$, then there exists no synchronizing
$n$-automaton with two input letters whose \rl\ is greater than
$n^2-3n+4$ but less than $(n-1)^2$.

\emph{(d)} If $n>7$ and $n$ is odd, then up to isomorphism there
exists exactly one synchronizing $n$-automaton with two input
letters and \rl\ $n^2-3n+4$, namely, the automaton
$\mathrsfs{D}'_n$, exactly two synchronizing $n$-automata with two
input letters and \rl\ $n^2-3n+3$, namely, the automata
$\mathrsfs{W}_n$ and $\mathrsfs{F}_n$, and exactly three
synchronizing $n$-automata with two input letters and \rl\
$n^2-3n+2$, namely, the automata $\mathrsfs{E}_n$,
$\mathrsfs{D}''_n$, and $\mathrsfs{B}_n$. There exists no
synchronizing $n$-automaton with two input letters whose \rl\ is
greater than $n^2-4n+7$ but less than $n^2-3n+2$.

\emph{(e)} If $n>8$ and $n$ is even, then up to isomorphism there exists exactly one synchronizing $n$-automaton with two input letters and
\rl\ $n^2-3n+4$,  namely, the automaton $\mathrsfs{D}'_n$, exactly one synchronizing $n$-automaton with two input letters and \rl\
$n^2-3n+3$, namely, the automaton $\mathrsfs{W}_n$, and exactly two synchronizing $n$-automata with two input letters and \rl\ $n^2-3n+2$,
namely, the automata $\mathrsfs{E}_n$ and $\mathrsfs{D}''_n$. There exists no synchronizing $n$-automaton with two input letters whose \rl\
is greater than $n^2-4n+6$ but less than $n^2-3n+2$.
\end{conjecture}

We also formulate a more special conjecture that can be treated as a quantitative form of the Road Coloring Conjecture mentioned in
Section~3. Since now we know that every primitive digraph has a synchronizing coloring~\cite{Tr09}, the notion of \rl\ can be naturally
extended to primitive digraphs. Namely, we call the \emph{\rl} of a primitive digraph the minimum length of \sws\ for all synchronizing
colorings of the digraph. This immediately leads to the question of how the \rl\ of a primitive digraph depends on the vertex number.

We notice that the digraphs of slowly \sa\ may admit colorings
with small \rl. Fig.~\ref{fig:cerny} illustrates this remark: the
first coloring of the digraph shown in the left is the \v{C}ern\'y
automaton $\mathrsfs{C}_4$ whose shortest \sw\ has length~9 while
the second coloring can be reset by the word $a^3$ of length~3. In
this connection, the series $W_n$ is of interest. In this series
each digraph has a unique (up to isomorphism) coloring. Therefore
the \rl\ of this coloring found in Theorem~\ref{theorem:anan}
coincides with the \rl\ of the digraph $W_n$ and provides a lower
bound for the problem under consideration. We conjecture that this
bound is in fact tight.

\begin{conjecture}
\label{hybrid} The \rl\ of every primitive $n$-digraph does not
exceed $n^2-3n+3$. If $n>3$, then up to isomorphism there exists
exactly one primitive $n$-digraph with \rl\ $n^2-3n+3$, namely,
the digraph $W_n$.
\end{conjecture}

Conjecture~\ref{hybrid} has been presented in several talks of the
second author since 2008 and some partial results towards its
proof have already been published,
see~\cite{Carpi&D'Alessandro:2010,Steinberg:2011}. It is clear
that Conjecture~\ref{hybrid} can be made more precise in the
flavor of Conjecture~\ref{general}: for instance, it is likely
that the digraph $D_n$ is the only (up to isomorphism) primitive
$n$-digraph with \rl\ $n^2-3n+2$, etc.

\paragraph*{5.2. The role of the alphabet size.} In our experiments
we restrict ourselves to automata with two input letters. This restriction is caused by the fact that an increase in the alphabet size
influences the number of automata much stronger than an increase in the state size. Therefore an exhaustive search through all automata
with more than two input letters is far beyond our computational capacities even  for automata with a modest number of states.
Table~\ref{table:number of automata} illustrates this fact. (The data in the table are calculated via a formula from~\cite{Lis69}.)

\begin{table}[ht]
\extrarowheight=2pt \tabcolsep=3.45pt  \caption{The number of
initially-connected automata with 2 and 3 input letters}
\label{table:number of automata} {\footnotesize
\begin{tabular}{|p{2.5cm}||c|c|c|} \hline
\# of states & 7 & 8 & 9  \\
\hline 2 input letters & \hfill 256\,182\,290 & \hfill 12\,665\,445\,248 & \hfill 705\,068\,085\,303\\
\hline 3 input letters &  500\,750\,172\,337\,212 &
572\,879\,126\,392\,178\,688 & 835\,007\,874\,759\,393\,878\,655 \\
\hline
\end{tabular}}
\end{table}

Nevertheless, there are some reasons to expect that the behavior of the function we are interested in (the number of \sa\ with a fixed
number of states as a function of \rl) does not heavily depend on the alphabet size. For instance, Trahtman's experiments whose results
were reported in~\cite{Tr06,Tr06a} revealed no 7-automaton with 3 or 4 input letters and with \rl\ larger than 32 and smaller than 36.
Thus, the value of the gap between the maximum and the next to maximum possible value of \rl\ is the same as for 7-automata with two input
letters.

We mention also the observation which, as far as we know, was first made in~\cite{Be10}: if there exists an upper bound of the form
$O(n^2)$ for the \rl\ of synchronizing $n$-automata with two input letters, then a bound of the same magnitude (but probably with a worse
constant) exists also for the \rl\ of synchronizing $n$-automata with any fixed size of the input alphabet.


\begin{thebibliography}{99}
\bibitem{AGW}
\textsl{R.\,L.\,Adler, L.\,W.\,Goodwyn, B.\,Weiss}, Equivalence of topological
Markov shifts, Israel J. Math., 27 (1977), 49--63.

\bibitem{AMR}
\textsl{M.\,Almeida, N.\,Moreira, R.\,Reis}, Enumeration and generation with a
string automata representation, Theor.\ Comput.\ Sci., 387 (2007), 93--102.

\bibitem{AGV10}
\textsl{D.\,S.\,Ananichev, V.\,V.\,Gusev, M.\,V.\,Volkov}, Slowly synchronizing
automata and digraphs, Mathematical Foundations of Computer Science, Lect.\
Notes Comput.\ Sci., 6281 (2010), 55--64.

\bibitem{AVZ}
\textsl{D.\,S.\,Ananichev, M.\,V.\,Volkov, Yu.\,I.\,Zaks}, Synchronizing
automata with a letter of deficiency 2, Theor.\ Comput.\ Sci., 376 (2007),
30--41.

\bibitem{Be10}
\textsl{M.\,Berlinkov}, Approximating the minimum length of synchronizing words
is hard,  Computer Science in Russia, Lect.\ Notes Comput.\ Sci., 6072 (2010),
37--47.

\bibitem{Be11}
\textsl{M.\,Berlinkov}, On a conjecture by Carpi and D'Alessandro, Int.\ J.
Foundations Comp.\ Sci., 22 (2011), 1565--1576.

\bibitem{Br}
\textsl{R.\,Brualdi, H.\,Ryser}, Combinatorial Matrix Theory, Cambridge University Press, 1991.

\bibitem{Carpi&D'Alessandro:2010}
\textsl{A.\,Carpi, F.\,D'Alessandro}, Independent sets of words and the synchronization problem, Adv. Appl. Math. 50 (2013) 339--355.

\bibitem{Ce64}
\textsl{J.\,\v{C}ern\'{y}}, Pozn\'{a}mka k homog\'{e}nnym
eksperimentom s kone\v{c}n\'{y}mi automatami,
Matematicko-fyzikalny \v{C}asopis Slovensk.\ Akad.\ Vied, 14, no.3
(1964), 208--216 (in Slovak).

\bibitem{DM62}
\textsl{A.\,L.\,Dulmage, N.\,S.\,Mendelsohn}, The exponent of a primitive
matrix, Can.\ Math.\ Bull., 5 (1962), 241--244.

\bibitem{DM64}
\textsl{A.\,L.\,Dulmage, N.\,S.\,Mendelsohn}, Gaps in the exponent set of
primitive matrices, Ill.\ J. Math., 8 (1964), 642--656.

\bibitem{Gusev:2011}
\textsl{V.\,Gusev}, Lower bounds for the length of reset words in {E}ulerian
automata, Reachability Problems, Lect.\ Notes Comput.\ Sci., 6945 (2011),
180--190.

\bibitem{Ka01}
\textsl{J.\,Kari}, A counter example to a conjecture concerning synchronizing
words in finite automata, Bull.\ European Assoc.\ Theor.\ Comput.\ Sci., 73
(2001), 146.

\bibitem{KoSz13}
\textsl{J.\,Kowalski, M.\,Szyku\l{}a}, The \v{C}ern\'y conjecture for small automata: experimental report,
\url{http://arxiv.org/abs/1301.2092}.

\bibitem{Lis69}
\textsl{V.\,A.\,Liskovets}, The number of connected initial automata, Kibernetika, no.3 (1969), 16--19 (in Russian).

\bibitem{OU10}
\textsl{J.\,Olschewski, M.\,Ummels}, The complexity of finding reset words in
finite automata, Mathematical Foundations of Computer Science, Lect.\ Notes
Comput.\ Sci., 6281 (2010), 568--579.

\bibitem{Pi83}
\textsl{J.-E.\,Pin}, On two combinatorial problems arising from automata
theory, Ann.\ Discrete Math., 17 (1983), 535--548.

\bibitem{RaAl05}
\textsl{J.\,L.\,Ram\'{\i}rez Alfons\'{\i}n}, The diophantine Frobenius problem,
Oxford University Press, 2005.

\bibitem{ST00}
\textsl{V.\,N.\,Sachkov, V.\,E.\,Tarakanov}, Combinatorics of
Nonnegative Matrices. American Mathematical Society, 2002.

\bibitem{Sa05}
\textsl{S.\,Sandberg}, Homing and synchronizing sequences, Model-Based Testing
of Reactive Systems, Lect.\ Notes Comput.\ Sci., 3472 (2005), 5--33.

\bibitem{ST11}
\textsl{E.\,Skvortsov, E.\,Tipikin}, Experimental study of the shortest reset
word of random automata, Implementation and Application  of Automata, Lect.\
Notes Comput.\ Sci., 6807 (2011), 290--298.

\bibitem{Steinberg:2011}
\textsl{B.\,Steinberg}, The \v{C}ern\'y conjecture for one-cluster automata
with prime length cycle, Theor.\ Comput.\ Sci., 412 (2011) 5487--5491.

\bibitem{Tr06}
\textsl{A.\,N.\,Trahtman}, An efficient algorithm finds noticeable trends and
examples concerning the \v{C}ern\'y conjecture, Mathematical Foundations of
Computer Science, Lect.\ Notes Comput.\ Sci., 4162 (2006), 789--800.

\bibitem{Tr06a}
\textsl{A.\,N.\,Trahtman}, Notable trends concerning the synchronization of
graphs and automata, Electr.\ Notes Discrete Math., 25 (2006), 173--175.

\bibitem{Tr09}
\textsl{A.\,N.\,Trahtman}, The Road Coloring Problem, Israel J. Math., 172
(2009), 51--60.

\bibitem{Tr11}
\textsl{A.\,N.\,Trahtman}, Modifying the upper bound on the length of minimal
synchronizing word,  Fundamentals of Computation Theory, Lect.\ Notes Comput.\
Sci, 6914 (2011), 173--180.

\bibitem{Vo08}
\textsl{M.\,V.\,Volkov}, Synchronizing automata and the \v{C}ern\'{y}
conjecture, Languages and Automata: Theory and Applications, Lect.\ Notes
Comput.\ Sci., 5196, (2008), 11--27.

\bibitem{Vo09}
\textsl{M.\,V.\,Volkov}, Synchronizing automata preserving a chain of partial
orders, Theor.\ Comput.\ Sci., 410 (2009), 3513--3519.

\bibitem{Wi50}
\textsl{H.\,Wielandt}, Unzerlegbare, nicht negative Matrizen, Math.\ Z., 52,
(1950), 642--648 (in German).
\end{thebibliography}
\end{document}